\pdfoutput=1 
\documentclass[11pt]{amsart}
\usepackage{fullpage}
\usepackage[pagebackref, colorlinks = true, linkcolor = blue, urlcolor  = blue, citecolor = red, bookmarks, hypertexnames=false]{hyperref}
\usepackage{amssymb}
\usepackage{amsmath}
\usepackage{amsthm}
\usepackage{graphicx,color,colordvi}
\usepackage{bbm}
\usepackage{varioref}
\usepackage[capitalise]{cleveref}
\usepackage{stmaryrd}
\usepackage[utf8]{inputenc}
\usepackage{dsfont}
\usepackage{mathtools}
\usepackage{wrapfig}
\usepackage[english]{babel}
\usepackage{physics}
\usepackage{braket}
\usepackage[table]{xcolor}
\usepackage[center]{caption}
\usepackage{tikz}
\usepackage{ifthen}
\usepackage{pgfplots}
\pgfplotsset{compat=1.9}
\usetikzlibrary{shapes,arrows.meta}
\usetikzlibrary{positioning}
\usetikzlibrary{shapes.geometric}
\RequirePackage[framemethod=default]{mdframed}
\usepackage[margin=1in]{geometry}
\usepackage{comment}
\usepackage{url}
\usepackage{ upgreek }
\usepackage{makecell}
\usepackage{fancyref} 
\usepackage{float}
\usepackage{enumitem}
\usepackage{autonum}
\usepackage{subcaption}
\newfloat{algorithm}{t}{lop}
\usepackage{array,rotating ,makecell, multirow, tabularx}
\usepackage{scrextend}
\usepackage[normalem]{ulem}

\newtheoremstyle{newdefinition}{}{}{\normalfont}{}{\bfseries}{}{\newline}
{\thmname{#1} \thmnumber{#2}\thmnote{ (#3)}}

\newtheoremstyle{newplain}{}{}{\itshape}{}{\bfseries}{}{1em}
{\thmname{#1} \thmnumber{#2}\thmnote{ (#3)}}

\newtheoremstyle{newremark}{}{}{\normalfont}{}{\bfseries}{}{1em}
{\thmname{#1}}

\theoremstyle{newdefinition}
\newtheorem{definition}{Definition}[section]

\theoremstyle{newplain}
\newtheorem{theorem}[definition]{Theorem}
\newtheorem{lemma}[definition]{Lemma}
\newtheorem{proposition}[definition]{Proposition}
\newtheorem{corollary}[definition]{Corollary}
\newtheorem{remark}[definition]{Remark}
\newtheorem{example}[definition]{Example}

\newtheoremstyle{myplain}{5pt}{5pt}{\itshape}{0pt}{\bfseries}{}{5pt plus 1pt minus 1pt}{}
\theoremstyle{myplain}
\newtheorem*{theorem*}{Theorem}
\newtheorem*{corollary*}{Corollary}

\DeclareMathOperator{\cH}{\mathcal{H}}
\DeclareMathOperator{\HH}{\mathcal{H}}
\DeclareMathOperator{\cS}{\mathcal{S}}
\DeclareMathOperator{\cK}{\mathcal{K}}

\DeclareMathOperator{\cB}{\mathcal{B}}

\DeclareMathOperator{\supp}{supp}

\DeclareMathOperator{\id}{id}

\def\Block[#1,#2,#3,#4]{

\def\r{0.3};

\ifthenelse{\NOT #4=0}{
\fill [#2] (-0.5,-0.5) rectangle ({#1-0.5},0.5);
}

\foreach \n in {1,...,#1}{ 

\shade[shading=ball, ball color=orange] ({\n-1},0) circle (\r);

}

\begin{scope}[decoration={brace,mirror,amplitude=7}]

\ifthenelse{#4=1}{

\draw [decorate] (-0.5,-0.6) --node[below=3mm]{$#3$} ({#1-0.5},-0.6);

}

\ifthenelse{#4=2}{
\draw [decorate] ({#1-0.5},0.6) --node[above=3mm]{$#3$} (-0.5,0.6);
}

\end{scope}
}

\begin{document}

\author{Andreas Bluhm}
\email{andreas.bluhm@univ-grenoble-alpes.fr}
\address{Univ.\ Grenoble Alpes, CNRS, Grenoble INP, LIG, 38000 Grenoble, France}
\author{Ángela Capel}
\email{ac2722@cam.ac.uk}
\address{Centre for Mathematical Sciences, Wilberforce Road, Cambridge CB3 0WA, United Kingdom
\newline \and Fachbereich Mathematik, Auf der Morgenstelle 10, 72076 Tübingen, Germany}
\author{Pablo Costa Rico}
\email{pablo.costa@tum.de}
\address{Department of Mathematics, Technische Universität München, Germany \newline \and Munich Center for Quantum Science and Technology (MCQST), Germany}
\author{Anna Jen\v{c}ov\'a}
\email{jenca@mat.savba.sk}
\address{Mathematical Institute, Slovak Academy of Sciences, Bratislava, Slovakia}

\title{Belavkin-Staszewski Quantum Markov Chains}

\date{\today}

\begin{abstract}
    It is well-known that the conditional mutual information of a quantum state is zero if, and only if, the quantum state is a quantum Markov chain. Replacing the Umegaki relative entropy in the definition of the conditional mutual information by the Belavkin-Staszewski (BS) relative entropy, we obtain the BS-conditional mutual information, and we call the states with zero BS-conditional mutual information Belavkin-Staszewski quantum Markov chains. In this article, we establish a correspondence which  relates quantum Markov chains and BS-quantum Markov chains. This correspondence allows us to find a recovery map for the BS-entropy in the spirit of the Petz recovery map. Furthermore, we show that, over the set of BS-quantum Markov chains, this correspondence constitutes an entanglement-breaking map. Moreover, we prove a structural decomposition of the Belavkin-Staszewski quantum Markov chains and also study states for which the BS-conditional mutual information is only approximately zero. We subsequently extend the aforementioned correspondence, structural decomposition and recovery map to arbitrary pairs of states and conditional expectations.  As an application of the correspondence, we find the first family of states with non-vanishing conditional mutual information for which it decays superexponentially fast with the size of the middle system. 
\end{abstract}

\maketitle

\section{Introduction}

In quantum information theory, the conditional mutual information associated to a quantum state $\rho_{ABC}$ is a measure of conditional independence between the systems $A$ and $C$ with respect to system $B$. A case of special interest is when the conditional mutual information is equal to zero. These states are known as quantum Markov chains \cite{Cocycles} and their structure was fully characterized in \cite{HaydenJozsaPetzWinter-StrongSubadditivity-2004}.  Quantum Markov chains  have found various applications in quantum information theory.  As an example, it was shown in  \cite{brown-2012} that quantum Markov chains are in one-to-one correspondence to Gibbs states of local, commuting Hamiltonians, which is a generalization of the classical Hammersley-Clifford theorem \cite{Hammersley1971MarkovFO}.  They are also states that satisfy the data-processing inequality for the relative entropy between the states $\rho_{ABC}$ and $\rho_{AB} \otimes \tau_C$, where $\tau_C=I_C/d_C$ is the maximally mixed state, and the partial trace on system $A$ as quantum channel with equality \cite{Petz2003}. 

The Umegaki relative entropy is not the only quantity that extends the classical Kullback-Leibler divergence to the quantum realm. Another possible extension, for which a data-processing inequality also holds true, is the Belavkin-Staszewski relative entropy (BS-entropy), which is an upper bound on the former. The BS-entropy has turned out to be a useful tool in quantum information theory, for example, for quantum channel discrimination \cite{fang2021geometric}. The BS-mutual information and BS-conditional mutual information have been used to study the decay of correlations in quantum spin chains with local, finite-range, translation-invariant interactions \cite{BluhmCapelPerezHernandez-ExpDecayMI-2021, gondolf2024conditional}.

While it holds that, for a quantum Markov chain $\rho_{ABC}$, the BS-conditional mutual information is zero, the converse is not true. In analogy to the quantum Markov chains, we call the set of states with zero BS-conditional mutual information \textit{Belavkin Staszewski quantum Markov chains}. In fact, it is known that there are quantum states that satisfy the data-processing inequality for the BS-entropy with equality, but which do not give equality in the data-processing inequality for the Umegaki relative entropy \cite{HiaiMosonyi-f-divergences-2017,Jencova2009}. These examples can be used to construct BS-quantum Markov chains that are not quantum Markov chains.  In this article, we provide  explicit examples of  such states  in Example \ref{example:BSqmcNotQmc}  and in Proposition \ref{prop:EntangledBSQMC}. Of great importance will be those of Proposition \ref{prop:EntangledBSQMC} since they contain a subfamily  that  satisfies that  the marginal $\rho_{AC}$ is entangled, a phenomenon that does not occur with quantum Markov chains. Therefore, this provides a substantial difference between quantum Markov chains and BS-quantum Markov chains.

Although the set of BS-quantum Markov chains is strictly larger than the set of quantum Markov chains, we show that, for every BS-quantum Markov chain $\rho_{ABC}$, it is possible to construct  a quantum Markov chain $\eta_{ABC}$ with $\eta_B=\tau_B$. This is proven in \cref{theo:StructureBSDPI}. Actually, the converse  also holds true, namely for every quantum Markov chain $\eta_{ABC}$ with $\eta_{B}=\tau_B$ we can construct a family of BS-quantum Markov chains (see \cref{remark:ConverseEta}). This correspondence leads to a number of interesting consequences. In \cite[Section 4]{gondolf2024conditional}, a completely positive linear map $\Phi_{B \to AB}$,  which satisfies that every BS-quantum Markov chain could be recovered by it, was defined. Here, we prove that the converse also holds, i.e., that any state that is recovered by $\Phi_{B \to AB}$ is indeed a BS-quantum Markov chain, proving in turn that $\Phi_{B \to AB}$ is a true recovery map. Even more, we also show that this recovery map can be rewritten in a very similar way to the Petz recovery map \cite{Petz-SufficiencyChannels-1978, Petz-SufficientSubalgebras-1986, Petz2003}. As a second remarkable consequence, we are able to find the structural decomposition of BS-quantum Markov chains in Theorem \ref{theo:StructureBSDPI}. From this structure theorem, we can then identify in Proposition \ref{prop:QMC_in_BS}  the quantum Markov chains as a subset of the BS-quantum Markov chains.

In Section \ref{subsec:CorrespondenceEtaOmega} we extend this correspondence via an $\overline{\eta}$ from triples $(\rho,\sigma, \mathcal{T})$ saturating the BS-data-processing inequality, where $\rho,\sigma$ are quantum states and $\mathcal{T}$ is a quantum channel with Stinespring representation $\mathcal{T}=\tr_E[V\cdot V^*]$ where $V:\cH \to \cK \otimes \mathcal{H}_E$ is a partial isometry, to triples $(\eta_{\rho},\eta_{\sigma},\tr_E)$ that saturate the data-processing inequality for the relative entropy under the constraint that $\tr_E[\nu]$ is the maximally mixed state. Conversely, given  triples $(\mu,\nu,\tr_E)$ that saturate the data-processing inequality for the relative entropy with the constraint that $\tr_E(\nu)$ is the maximally mixed state, we can create the inverse correspondence $\omega(X,V)$ where $V$ is a partial isometry and $X$ lies in the image of $\tr_E$, i.e. we can define states $\omega_{\mu}(X,V)$ and $\omega_{\nu}(X,V)$ such that the triple $(\omega_{\mu}(X,V),\omega_{\nu}(X,V),\mathcal{T})$ saturates the data-processing inequality for the BS-relative entropy. This correspondence is shown in Figure \ref{fig:3} and proved in Theorem \ref{thm:bs_petz} and  in Corollary \ref{coro:Petz_BS_converse}. Moreover, Corollary \ref{coro:recoveryPhiChannels} provides a recovery map in the spirit of the Petz recovery map for the BS-relative entropy which is linear, completely positive  but not trace preserving, which generalizes  the map $\Phi_{B\to AB}$ to the general case.

\begin{figure}[ht]
\begin{center}

\includegraphics[scale=0.25]{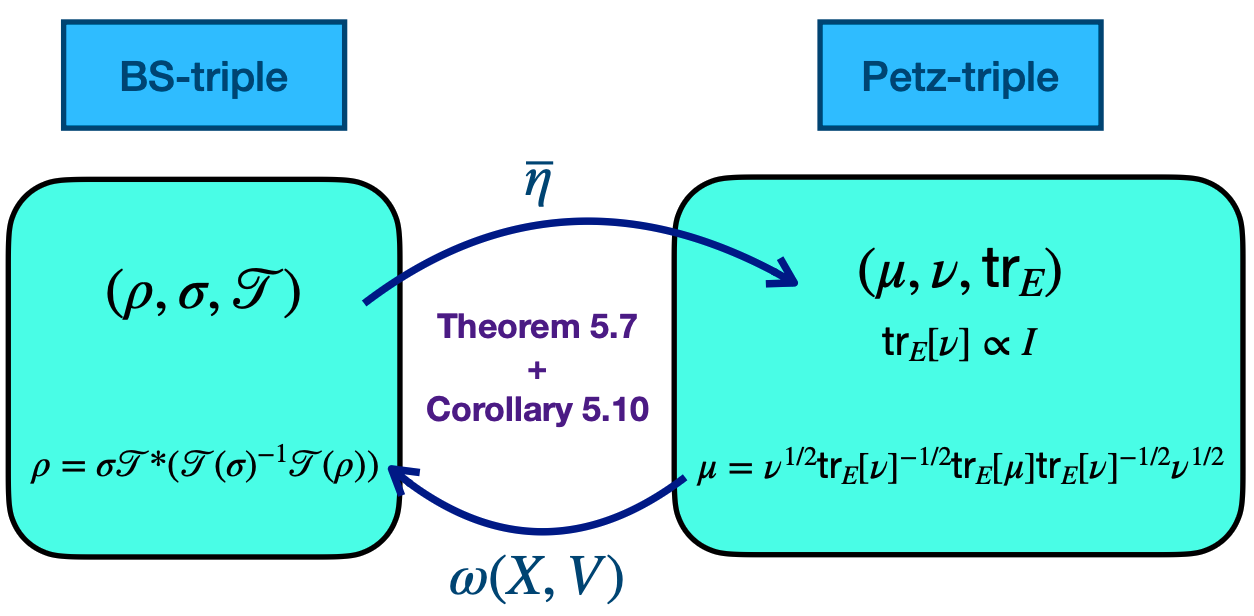}

  \caption{Correspondence between states and channels that saturate the DPI for the BS-entropy, and states and partial traces that saturate the DPI for the relative entropy. The map to the right, $\overline{\eta}$, yields a unique Petz-triple given a BS-triple, whereas the map to the left, $\omega(X,V)$, allows to create multiple BS-triples from a Petz-triple.} 
  \label{fig:3}
      
\end{center}
\end{figure}

In addition, in Section \ref{subsec:StructuralDecomposition} we show how to find the structural decomposition of the states saturating the BS-data-processing inequality using the structure of the states that saturate the data-processing inequality for the relative entropy, which we obtain employing tools developed  in \cite{jenvcova2006sufficiency} and \cite{mosonyi2004structure}. Conversely, we also show how to obtain directly the decomposition of states saturating the BS data-processing inequality in Theorem \ref{theo:StructureBSDPI_appendix} and we use it to find the structural decomposition of the states that saturate the data-processing inequality for the  relative entropy for conditional expectations using the previous correspondence.

In the case of quantum Markov chains, the interest has not been limited to exact quantum Markov chains, but also approximate versions have been considered, i.e., states with $I_\rho(A:C|B) \leq \varepsilon$. We study approximate BS-quantum Markov chains in Propositions \ref{prop:relation_Irev_I} and \ref{prop:appox-rotated2}, where we relate them to approximate quantum Markov chains. Finally, in Theorem \ref{thm:superexponential_decay}, we use these results to show that whenever $\rho_{ABC}$ is the Gibbs state of a quantum spin chain with local, finite-range, translation-invariant interactions at any positive temperature, then the associated $\eta_{ABC}$ (not necessarily a quantum Markov chain) has superexponentially-decaying conditional mutual information, building up on \cite{gondolf2024conditional}. Furthermore, we show that our reversed BS-conditional mutual information constitutes an upper bound (up to prefactors) to a quantity that, when exponentially-decaying with the distance between $A$ and $C$ for a state $\rho_{ABCD}$, gives a positive spectral gap for the Davies Lindbladian associated with unique fixed point $\rho_{ABCD}$.

\section{Notation and preliminaries}

\subsection{Relative entropies}

Let us consider a finite-dimensional Hilbert space $\HH$ and let $ \rho, \sigma \in \mathcal{S}(\HH)$ be two quantum states on it. Their \textit{Umegaki relative entropy}
\cite{Umegaki-RelativeEntropy-1962} (or just relative entropy for short, since it is the relative entropy customarily used in quantum information theory) is defined as 
\begin{equation}\label{eq:eq_relative_entropy}
    D(\rho \Vert \sigma) := \begin{cases}
        \tr[\rho \log\rho - \rho \log \sigma] & \text{if } \ker \sigma \subseteq \ker \rho \, ,\\
        + \infty & \text{otherwise} \, ,
    \end{cases}
\end{equation}
and their \textit{Belavkin-Staszewski (BS) entropy} \cite{BelavkinStaszewski-BSentropy-1982} by 
\begin{equation}
    \widehat{D}(\rho \Vert \sigma) := \begin{cases}
        \tr[\rho\log \rho^{1/2} \sigma^{-1} \rho^{1/2}] & \text{if } \ker \sigma \subseteq \ker \rho \, , \\
        + \infty & \text{otherwise} \, .
    \end{cases}
\end{equation}
In the case of $\rho$ and $\sigma$ commuting, the two entropies coincide.  Otherwise, the BS-entropy is strictly larger than the relative entropy \cite{HiaiMosonyi-f-divergences-2017}.

Both notions above constitute quantum generalizations of the classical Kullback-Leibler divergence. The Umegaki relative entropy  between two quantum states measures their distinguishability \cite{hiai1991proper}. Moreover, after the application of a quantum channel, i.e. a completely positive and trace-preserving linear map $\mathcal{T}: \mathcal{S}(\mathcal{H}) \rightarrow  \mathcal{S}(\mathcal{K})$,  the distinguishability between those states can never increase. This phenomenon is called \textit{data-processing inequality} \cite{Petz2003}:
\begin{equation}\label{eq:DPIrelativeentropy}
    D(\rho \| \sigma ) \geq D(\mathcal{T} (\rho) \| \mathcal{T}(\sigma) )  \, .
\end{equation}
However, there are situations in which, after the application of a quantum channel, the Umegaki relative entropy does not decrease. This \textit{saturation} of the data-processing inequality was studied by Petz in \cite{Petz-SufficiencyChannels-1978, Petz-SufficientSubalgebras-1986, Petz2003}, where he proved:
\begin{equation}\label{eq:PetzRecoveryMap}
    D(\rho \| \sigma ) = D(\mathcal{T} (\rho) \| \mathcal{T}(\sigma) )  \; \Leftrightarrow \; \rho = \sigma^{1/2} \mathcal{T}^* (\mathcal{T}(\sigma)^{-1/2} \mathcal{T}(\rho) \mathcal{T}(\sigma)^{-1/2}) \sigma^{1/2} \, .
\end{equation}
Note that the map applied to $\mathcal T(\rho)$ on the right hand side is a quantum channel. It is called the \textit{Petz recovery map} and we denote it hereafter by $\mathcal{P}^\sigma_{\mathcal{T}}$, i.e., 
\begin{equation}
   \mathcal{P}^\sigma_{\mathcal{T}}(X) =  \sigma^{1/2} \mathcal{T}^* (\mathcal{T}(\sigma)^{-1/2} X \mathcal{T}(\sigma)^{-1/2}) \sigma^{1/2} \, , \qquad \forall X \in \mathcal S(\mathcal K).
\end{equation}
Eq. \eqref{eq:PetzRecoveryMap} then reads as an equivalence between saturation of the DPI for the relative entropy and $\rho$ being a fixed point of $\mathcal{P}^\sigma_{\mathcal{T}}\circ \mathcal{T}$. This inequality has been strengthened multiple times by providing lower non-negative bounds on the difference between the LHS and the RHS of Eq.\ \eqref{eq:DPIrelativeentropy} in terms of various measures of the `distance' from a state $\rho$ to its Petz recovery map (or to a rotated version of it) \cite{Fawzi2015,Sutter2017b,Junge2018}, e.g.,  \cite{CarlenVershynina-Stability-DPI-RE-2017}
\begin{equation} \label{eq:carlen-vershynina}
     D(\rho \| \sigma ) - D(\mathcal{N} (\rho) \| \mathcal{N}(\sigma) ) \geq \left( \frac{\pi}{8} \right)^4 \norm{\rho^{-1}}^{-2} \norm{\mathcal{N}(\sigma)^{-1}}^{-2} \norm{ \mathcal{P}^\sigma_{\mathcal{N}}\circ \mathcal{N}(\rho) - \rho }_1^4 \, .
\end{equation}
for $\mathcal{N}$ a conditional expectation.

Let us move now to the setting of the BS-entropy. The data-processing inequality also holds for this quantity, namely for every $\rho, \sigma \in \mathcal{S}(\HH)$ and every quantum channel $\mathcal{T}: \mathcal{S}(\HH) \rightarrow \mathcal{S}(\mathcal{K})$, we have
\begin{equation}
    \widehat{D}(\rho \| \sigma) \geq \widehat{D} (\mathcal{T}(\rho) \| \mathcal{T}(\sigma)) \, .
\end{equation}
Additionally, saturation of the BS-entropy was proven in \cite{BluhmCapel-BSentropy-2019} to be equivalent to 
\begin{equation}
    \widehat{D}(\rho \| \sigma) = \widehat{D} (\mathcal{T}(\rho) \| \mathcal{T}(\sigma)) \, \Leftrightarrow \, \rho = \sigma \mathcal{T}^*(\mathcal{T}(\sigma)^{-1} \mathcal{T}(\rho))  \, .
\end{equation}
Analogously to the introduction of the Petz recovery map, this equivalence motivated the definition of the so-called \textit{BS-recovery condition} \cite{BluhmCapel-BSentropy-2019} in the following way:
\begin{equation}\label{eq:BS-recovery-condition}
 \mathcal{B}^\sigma_{\mathcal{T}} (\cdot) := \sigma \mathcal{T}^*(\mathcal{T}(\sigma)^{-1} (\cdot) \, ) \, . 
\end{equation}
Along the lines of the strengthened DPI for the relative entropy recalled above, some authors of this article proved in \cite{BluhmCapel-BSentropy-2019} the following inequality:
\begin{equation}\label{eq:DPIBSentropy}
     \widehat{D}(\rho \| \sigma ) - \widehat{D}(\mathcal{T} (\rho) \| \mathcal{T}(\sigma) ) \geq \left( \frac{\pi}{8} \right)^4 \norm{\rho^{-1/2}\sigma \rho^{-1/2}}^{-4} \norm{\mathcal{T}(\rho)^{-1}}^{-2} \norm{ \mathcal{B}^\rho_{\mathcal{T}}\circ \mathcal{T}(\sigma) - \sigma }_2^4 \, .
\end{equation}

The map $\mathcal{B}^\sigma_{\mathcal{T}}$ is trace preserving but, unfortunately, is not positive or even Hermitian-preserving in general. To deal with this issue,  we will construct in Section \ref{sec:equality_conditions} a new recovery condition $\mathcal{B}^{\sigma,\text{sym}}_{\mathcal{T}}$ for the BS-entropy by symmetrizing the former one. The map $\mathcal{B}^{\sigma,\text{sym}}_{\mathcal{T}}$ is positive, but it is not linear.

\subsection{Conditional mutual informations}

Next, let us now consider a special case of the previous setting. Consider a tripartite Hilbert space $\mathcal{H}_{ABC} = \cH_A \otimes \cH_B \otimes \cH_C$ and  $\rho_{ABC} \in \cS_{+}(\cH_{ABC})$ a positive-definite state. On such multipartite systems, we will sometimes drop identity operators for readability, i.e., write $O_A$ instead of $O_A \otimes I_{BC}$ for $O_A \in \mathcal \mathcal{B}(\mathcal H_A)$ and $\mathcal T_{A \to X}$ instead of $\mathcal T_{A \to X} \otimes \id_{BC}$ for a map $\mathcal T_{A \to X}$ acting on $\mathcal \mathcal{B}(\mathcal H_A)$ and mapping it to $\mathcal \mathcal{B}(\mathcal H_X)$ for some Hilbert space $\mathcal H_X$.  For a partition $X, Y$ of the set $\{A,B,C\}$, we will denote  $\rho_{X}=\tr_{Y}\rho_{ABC}$, where $\tr_Y:\mathcal{B}(\mathcal{H}_X \otimes \mathcal{H}_Y) \to \mathcal{B}(\mathcal{H}_X)$ denotes the partial trace defined as the unique linear map satisfying that the condition $\tr[S(T\otimes I)]=\tr[\tr_Y(S)T]$ holds for every $T \in \mathcal{B}(\mathcal{H}_X)$ and $S \in \mathcal{B}(\mathcal{H}_X \otimes \mathcal{H}_Y)$. To shorten notation, we will also write sometimes $\rho_X$ instead of $ \rho_X \otimes I_Y$.

We define the \textit{conditional mutual information (CMI)} of $\rho_{ABC}$ between $A$ and $C$ conditioned on $B$ by
\begin{equation}
    I_{\rho}(A:C | B) := S(\rho_{AB}) +  S(\rho_{BC})- S(\rho_{ABC}) - S(\rho_{B}) \, ,
\end{equation}
for $S(\rho_X) := - \tr[\rho_X \log \rho_X]$ the von Neumann entropy of $\rho_X$ for $X \subseteq ABC$. 
The well-known property of strong subadditivity of the von Neumann entropy \cite{LiebRuskai-Subadditivity-1973} is equivalent to the non-negativity of the conditional mutual information. A state $\rho_{ABC}$ for which the CMI vanishes is called \textit{quantum Markov chain} (QMC). These states admit the following condition \cite{Petz2003,HaydenJozsaPetzWinter-StrongSubadditivity-2004}:
\begin{equation}\label{eq:PetzCondition}
    \rho_{ABC} = \rho_{AB}^{1/2} \rho_B^{-1/2} \rho_{BC}  \rho_B^{-1/2}  \rho_{AB}^{1/2} = :  (\mathcal{P}_{B\to AB}\otimes \id_C)(\rho_{BC}) \, ,
\end{equation}
Note that to ease notation, we have written here $\mathcal{P}_{B\to AB}$ for the Petz recovery map in the case of $\mathcal T = \tr_A$ and $\sigma=\rho_{AB}\otimes \tau_C$, where $\tau_C = I_C / d_C$. 

In the same setting, we can define the \textit{BS-conditional mutual information} (BS-CMI in short) of $\rho_{ABC}$ between $A$ and $C$ conditioned on $B$ in different forms, with the common ground that all vanish under the same conditions:

\begin{subequations}
\begin{align}\label{eq:BS-CMI1}
 \widehat{I}^{\mathrm{os}}_{\rho}(A;C | B)  := \widehat{D}(\rho_{ABC} \|  \rho_{AB} \otimes \tau_C  ) -\widehat{D}(\rho_{BC} \| \rho_{B} \otimes \tau_C   )  \, , 
 \end{align}
 \begin{align}\label{eq:BS-CMI2}
\widehat{I}^{\mathrm{ts}}_{\rho}(A;C | B)  := \widehat{D}(\rho_{ABC} \| \rho_{AB} \otimes \rho_C   ) -\widehat{D}(\rho_{BC} \| \rho_{B} \otimes \rho_C   )  \, , 
 \end{align}
 \begin{align}\label{eq:BS-CMI3}
\widehat{I}^{\mathrm{rev}}_{\rho}(A;C | B)  := \widehat{D}(\rho_{AB} \otimes \tau_C   \| \rho_{ABC}  ) -\widehat{D}(\rho_{B} \otimes \tau_C  \| \rho_{BC}  )  \, . 
 \end{align}
\end{subequations}
Note that contrary to the CMI, which is written with a double colon, we write a semicolon here, since the different versions of BS-CMI are not symmetric in $A$ and $C$. These notions were introduced in \cite{gondolf2024conditional} exchanging the roles of $A$ and $C$ above, but here we are using the expressions presented above as we find them more intuitive. Note that for the CMI this does not actually change anything, as it is symmetric in $A$ and $C$. If we were to replace the BS-entropy by the Umegaki relative entropy above, the first two quantities would reduce to $I_\rho(A:C|B)$. Translating the previous conditions for saturation of DPI into this setting, we have for $x\in \{ \mathrm{os}, \mathrm{ts}, \mathrm{rev}\}$
\begin{align}\label{eq:BSCMIZero}
   \widehat{I}^x_{\rho}(A;C | B) = 0 \, & \Leftrightarrow \, \rho_{ABC} = \rho_{AB} \rho_B^{-1} \rho_{BC} =:  (\mathcal{B}_{B\to AB}\otimes \id_C)(\rho_{BC}) \\
   & \, \Leftrightarrow \, \rho_{ABC} = ( \rho_{AB} \rho_B^{-1} \rho_{BC}^2 \rho_B^{-1} \rho_{AB} )^{1/2}  =:  \mathcal{B}^{\text{sym}}_{B\to AB}(\rho_{BC}) \, .
\end{align}
Again, we have written $\mathcal{B}_{B\to AB}$ and $\mathcal{B}^{\text{sym}}_{B\to AB}$ for the respective recovery conditions in the special case $\mathcal T = \tr_A$ and $\sigma = \rho_{AB} \otimes \tau_C$. This equivalence will be shown in \Cref{theo:StructureBSDPI_appendix}. We call states that satisfy any of the conditions above \textit{BS-quantum Markov chains} (BS-QMC). The first condition has been used in the estimation of decay of correlations of Gibbs states of local, finite-range, translation-invariant 1D Hamiltonians at any positive temperature in the past \cite{BluhmCapelPerezHernandez-ExpDecayMI-2021}. In the recent paper \cite{gondolf2024conditional}, a reversed DPI based on the first equivalence has been used to show superexponential decay of the three BS-CMIs introduced above with the size of $|B|$, for Gibbs states of local, finite-range, translation-invariant 1D Hamiltonians at any positive temperature. Additionally, building on Eq.\ \eqref{eq:DPIBSentropy} for $\widehat{I}^{\mathrm{rev}}_{\rho}(A;C | B)$ and an additional technical lemma, the following inequality was derived in the same paper
\begin{equation}\label{ineq:LowerBoundReversedCMI}
   \widehat{I}^{\mathrm{rev}}_{\rho}(A;C | B) \geq  \left( \frac{\pi}{8} \right)^4 \|\rho_{BC}^{-1/2}\rho_{ABC}\rho_{BC}^{-1/2} \|_\infty^{-2}\|\Phi_{B\to BC}(\rho_{AB})-\rho_{ABC}\|_1^4 
\end{equation}
for the map
\begin{equation}\label{MapPhi}
    \Phi_{B\to BC}(X) = \rho_B^{1/2}(\rho_B^{-1/2}\rho_{BC}\rho_B^{-1/2})^{1/2}\rho_B^{-1/2}X\rho_B^{-1/2}(\rho_B^{-1/2}\rho_{BC}\rho_B^{-1/2})^{1/2}\rho_B^{1/2} \, . 
\end{equation}
As an immediate consequence of this inequality, we have that $ \widehat{I}^{\mathrm{rev}}_{\rho}(A;C | B) = 0$ implies $\rho_{ABC}= ( \id_A \otimes \Phi_{B\to BC} )(\rho_{AB}) $, but the converse was left as an open question in \cite{gondolf2024conditional}. We answer this question in the affirmative in this article in Corollary \ref{theo:recoveryBSDPI}.

\section{States saturating the data-processing inequality for the BS-entropy}\label{sec:CharacterizationEta}
\subsection{Structure of BS-quantum Markov chains}

Similarly to the way in which quantum Markov chains were introduced as states for which the CMI vanishes, we can analogously introduce BS-quantum Markov chains.

\begin{definition}[BS-Quantum Markov Chain]
A state $\rho_{ABC}\in \cS(\cH_{ABC})$ is said to be a \emph{BS-quantum Markov chain  (BS-QMC)} if  the BS-CMIs from \eqref{eq:BS-CMI1}, \eqref{eq:BS-CMI2}, \eqref{eq:BS-CMI3}  vanish i.e.
\begin{equation}
\rho_{ABC} \text{ BS-QMC }\Leftrightarrow \; I_\rho^x(A;C|B)= 0 \text{ for any } x \in \{ \operatorname{os}, \operatorname{ts}, \operatorname{rev}\} \, . 
\end{equation}
\end{definition}

\begin{remark}\label{rem:BS-QMC}
Note that $I_\rho^x(A;C|B)= 0 \text{ for any } x \in \{ \operatorname{os}, \operatorname{ts}, \operatorname{rev}\} \, $ is equivalent to all of them vanishing, as by \cite{BluhmCapel-BSentropy-2019} we have: 
\begin{subequations}
\begin{align}\label{eq:BS-QMC1}
 \widehat{I}^{\mathrm{os}}_{\rho}(A;C | B) =0  \quad \Leftrightarrow \quad \rho_{ABC} = \rho_{AB} \, \rho_B^{-1} \rho_{BC}    \, , 
 \end{align}
 \begin{align}\label{eq:BS-QMC2}
 \widehat{I}^{\mathrm{ts}}_{\rho}(A;C | B) =0  \quad \Leftrightarrow \quad \rho_{ABC} = \rho_{AB} \otimes \rho_C \, (  \rho_B \otimes \rho_C)^{-1}   \rho_{BC}    \, , 
 \end{align}
 \begin{align}\label{eq:BS-QMC3}
 \widehat{I}^{\mathrm{rev}}_{\rho}(A;C | B) =0  \quad \Leftrightarrow \quad \rho_{AB} = \rho_{ABC} \,  \rho_{BC}^{-1}  \rho_{B}    \, . 
 \end{align}
\end{subequations}

\end{remark}

A natural starting point for the comparison between BS-QMCs and QMCs is given by  their structural decomposition. For the latter case, this was studied in \cite{HaydenJozsaPetzWinter-StrongSubadditivity-2004}, obtaining that $\rho_{ABC}$ is a QMC between $A \leftrightarrow B \leftrightarrow C $ (meaning $I_\rho(A:C|B) = 0$) if, and only if, there exist Hilbert spaces $\mathcal{H}_{B_n^L}$, $\mathcal{H}_{B_n^R}$, and a unitary  $\displaystyle U_B:\mathcal{H}_B \to \oplus_{n=1}^N\left( \mathcal{H}_{B_n^L} \otimes \mathcal{H}_{B_n^R}\right)$,
 such that, with $\{p_n\}$ being a probability distribution and quantum states $\Tilde{\rho}_{AB_n^L} \in \mathcal S(\mathcal{H}_{AB_n^L})$ and $\Tilde{\rho}_{B_n^R C} \in \mathcal S(\mathcal{H}_{B_n^RC})$,
\begin{equation}\label{eq:StructuralQMC}
    \rho_{ABC} =U_B^*\left( \underset{n }{\bigoplus} \,p_n \,\Tilde{\rho}_{AB_n^L} \otimes  \Tilde{\rho}_{B_n^R C}\right)U_B \, .
\end{equation}
It turns out that a similar decomposition can be found for BS-QMC, as our first main result shows:

\begin{theorem}\label{theo:StructureBSDPI} Assume that $\rho_{ABC}\in \cS(\cH_{ABC})$ is such that $\rho_B$ is invertible. Define the state 
\begin{equation}\label{definition_eta}
\eta_{ABC}:=\frac{1}{d_B}\rho_B^{-1/2}\rho_{ABC}\rho_B^{-1/2}.
\end{equation}
Then, the following are equivalent:

\begin{enumerate}
\item[(i)] $\rho_{ABC}$ is a BS-QMC.
\item[(ii)] $\rho_{ABC}=\rho_{AB}\rho_B^{-1}\rho_{BC}$.
\item[(iii)] The marginals $\eta_{AB}$ and $\eta_{BC}$ commute, and we have
$\rho_{ABC}=d_B^2\rho_B^{1/2}\eta_{AB}\eta_{BC}\rho_{B}^{1/2}$.
\item[(iv)] $\eta_{ABC}$ is a QMC.
\item[(v)] There are Hilbert spaces  $\mathcal{H}_{B_n^L}$, $\mathcal{H}_{B_n^R}$ and a
unitary  $\displaystyle U_B:\mathcal{H}_B \to \bigoplus_{n=1}^N\left( \mathcal{H}_{B_n^L} \otimes \mathcal{H}_{B_n^R}\right)$,
 such that 
\begin{equation}\label{AnnaStates}
\rho_{ABC}=\rho_B^{1/2}U^*_B\left(\bigoplus_n d_B\, p_n \, \Tilde{\eta}_{AB_{n}^L}\otimes
\Tilde{\eta}_{B_n^RC}\right)U_B\rho_B^{1/2}
\end{equation}
for some states $\Tilde{\eta}_{AB_n^L}$ on $\mathcal {H}_{AB_n^L}$ and $\Tilde{\eta}_{B_n^RC}$ on $\mathcal {H}_{B_n^RC}$
and a probability distribution $\{p_n\}$.

\end{enumerate}
\end{theorem}

\begin{proof}

\noindent \underline{(i) $\iff$ (ii).} This equivalence  follows by Remark \ref{rem:BS-QMC}.

\vspace{0.1cm}

\noindent \underline{(ii) $\Rightarrow$ (iii).} If (ii) holds, then 
clearly $\rho_{ABC}=d_B^2\rho_B^{1/2}\eta_{AB}\eta_{BC}\rho_B^{1/2}=\rho_{ABC}^*$.
Since $\rho_B$ is invertible, this implies that 
$[\eta_{AB},\eta_{BC}]=0$, so that (iii) holds. 

\vspace{0.1cm}

\noindent \underline{(iii) $\Rightarrow$ (iv).} Assume (iii), then since $\eta_B=d_B^{-1}I_B=\tau_B$, we obtain 
\[
\eta_{ABC}=d_B  \eta_{AB}  \eta_{BC}=\eta_{AB}^{1/2}\eta_B^{-1/2}\eta_{BC}\eta_B^{-1/2}\eta_{AB}^{1/2}
\]
so that $\eta_{ABC}$ is a QMC.

\vspace{0.1cm}

\noindent \underline{(iv) $\Rightarrow$ (v).} If (iv) holds, then by the structural decomposition in \cite{HaydenJozsaPetzWinter-StrongSubadditivity-2004} there are Hilbert spaces $\mathcal{H}_{B_n^L}$, $\mathcal{H}_{B_n^R}$ and a
unitary  $U_B:\mathcal{H}_B \to \bigoplus_{n=1}^N\left( \mathcal{H}_{B_n^L} \otimes \mathcal{H}_{B_n^R}\right)$ such that
\[
\eta_{ABC}=U^*_B\left(\bigoplus_n p_n \, \Tilde{\eta}_{AB_{n}^L}\otimes
\Tilde{\eta}_{B_n^RC}\right)U_B,
\]
this proves (v). 

\vspace{0.1cm}

\noindent \underline{(v) $\Rightarrow$ (ii).}
Finally, suppose that (v) holds, then from
\[
\tau_B=d_B^{-1}\tr_{AC} \left[\rho_B^{-1/2}\rho_{ABC}\rho_B^{-1/2} \right]=U_B^*\left(\bigoplus_n p_n\, \Tilde{\eta}_{B_n^L}\otimes
\Tilde{\eta}_{B_n^R}\right)U_B
\]
we infer that $\Tilde{\eta}_{B_n^L}=\tau_{B_n^L}$ and $\Tilde{\eta}_{B_n^R}=\tau_{B_n^R}$. It follows that
$\rho_{AB}=\rho_B^{1/2}U^*_B\left(\bigoplus_n d_B \, p_n \, \Tilde{\eta}_{AB_n^L}\otimes
\tau_{B_n^R}\right)U_B\rho_B^{1/2}$ and
similarly $\rho_{BC}=\rho_B^{1/2}U^*_B\left(\bigoplus_n d_B \, p_n \, \tau_{B_n^L}\otimes
\Tilde{\eta}_{B_n^RC}\right)U_B\rho_B^{1/2}$. The condition (ii) is immediate from this.
\end{proof}

\begin{remark}
 Note that the assumption that $\rho_B$ is invertible can be easily removed. Indeed, let $P_B:= \supp(\rho_B)$. Since the supports 
 $\supp(\rho_{ABC})$, $\supp(\rho_{AB})$, $\supp(\rho_{BC})$ are all
contained in $P_B$, we may define the state
\[
\eta_{ABC}=\tr[P_B]^{-1}\rho_{B}^{-1/2}\rho_{ABC}\rho_{B}^{-1/2} \, .
\]
All the statements and proofs remain the same, except that now the marginal is $\eta_B=\tr[P_B]^{-1}P_B$, which commutes with all the states involved. Alternatively, we may always assume that $\rho_B$ is invertible by restricting to the subspace $P_B\cH_B$.

\end{remark}

\begin{corollary}
    Let $\rho_{ABC}$ be a BS-QMC with associated QMC $\eta_{ABC}$ as in Eq.\ \eqref{definition_eta}. Then, 
    \begin{center}
        $\widehat{I}_{\rho}(A;C\vert B)=0$ if, and only if, $I_{\eta}(A:C\vert B)=0 \, .$ 
    \end{center}
\end{corollary}
\begin{proof}
This follows directly from the equivalence of (i) and (iv) in Theorem \ref{theo:StructureBSDPI}.     
\end{proof}

\begin{remark}\label{remark:ConverseEta}
 Note that condition (iv) in Theorem \ref{theo:StructureBSDPI} shows that we can map any BS-QMC $\rho_{ABC}$ to a QMC $\eta_{ABC}$ satisfying $\eta_B=\tau_B$. 
 Condition (iii) shows the converse, i.e., we can map any  QMC $\eta_{ABC}$ with $\eta_B=\tau_B$ to a family of BS-QMCs by defining $\omega_{ABC}(X_B):=d_B^2 X_B^{1/2}\eta_{AB}\eta_{BC}X_{B}^{1/2}$, for any $X_B \in \cS(\cH_{B})$. We will then have $\rho_B=X_B$. As a special case, given the   QMC decomposition
 \begin{equation}
      \eta_{ABC} = U_B^*\left(\underset{n }{\bigoplus} p_n \, \Tilde{\eta}_{AB_n^L} \otimes  \Tilde{\eta}_{B_n^R C}\right)U_B \, ,
 \end{equation}
  if we take
 \begin{equation}
         X_B:=U_B^*\left(\underset{n }{\bigoplus} \frac{1}{d_{B_n}}\tilde{\rho}_{B_n^L}\otimes \tilde{\rho}_{B_n^R}  \right)U_B,
     \end{equation}
for any $\tilde{\rho}_{B_n^L} \in \cS(\mathcal{H}_{B_n^L})$, $\tilde{\rho}_{B_n^R} \in \cS(\mathcal{H}_{B_n^R})$, then $\omega_{ABC}(X_B)=d_B^2 X_B^{1/2}\eta_{AB}\eta_{BC}X_{B}^{1/2}$ is a QMC. We will discuss more about when a BS-QMC is a QMC in  Section \ref{BS-QMC-and-QMC}. 
\end{remark}
\begin{remark}\label{rem:Gibbs_comm_QMC}
 Consider a BS-QMC $\rho_{ABC}$ with associated QMC $\eta_{ABC}$. It is proven, e.g., in \cite{brown-2012}, that QMCs are Gibbs states of local commuting Hamiltonians, that is,  we can write
 \begin{equation}
     \eta_{ABC} =\frac{1}{Z_{ABC}} e^{-H_{AB}-H_{BC}} \, ,
 \end{equation}
 where $Z_{ABC}$ is the normalization constant, where $H_{AB} $ and $H_{BC}$ are Hermitian operators supported in $AB$ and $BC$, respectively, and such that  $[H_{AB}, H_{BC}] = 0$.  As a consequence, by Eq.\ \eqref{definition_eta}, we can write any BS-QMC as 
 \begin{equation}
     \rho_{ABC}= \frac{1}{Z_{ABC}}\rho_B^{1/2} e^{-H_{AB}-H_{BC}}\rho_B^{1/2} \, .
 \end{equation}
 
\end{remark}

\begin{remark}\label{rem:other-f-divergences}
   Theorem \ref{theo:StructureBSDPI}  describes not only the set of states with BS-CMI equal to zero,  but also the set of states with CMI equal to zero if we replace the BS-entropy in the definition of $\widehat{I}^x_\rho(A;C|B)$, $x \in \{\mathrm{os},\mathrm{ts},\mathrm{rev}\}$, by any maximal $f$-divergence  \cite{Matsumoto2018,HiaiMosonyi-f-divergences-2017}, with $f$ an operator convex but non-linear function. In fact, by \cite[Theorem 3.34]{HiaiMosonyi-f-divergences-2017} it is enough to check the saturation of the data-processing inequality in one maximal $f$-divergence defined by a  non-linear operator convex function.
\end{remark}

Condition (ii) in Theorem \ref{theo:StructureBSDPI} is an equality condition in terms of the recovery condition  $\mathcal{B}_{B \rightarrow AB}$. The map $\mathcal{B}_{B \rightarrow AB}$ is linear and trace preserving but not positive. Using the commutativity of the marginals of $\eta_{ABC}$ in condition (iii), we can obtain a new recovery condition  in the following way:
\begin{equation}\label{eq:FromBtoPhi}
    \begin{split}
        (\mathcal{B}_{B \rightarrow AB}\otimes \id_C)(\rho_{BC})&=\rho_{AB}\rho_B^{-1}\rho_{BC}\\
        &=\rho_{B}^{\frac{1}{2}}\underbrace{(\rho_B^{-\frac{1}{2}}\rho_{AB}\rho_B^{-\frac{1}{2}})}_{d_B \eta_{AB}}\underbrace{(\rho_B^{-\frac{
1}{2}}\rho_{BC}\rho_B^{-\frac{1}{2}})}_{d_B \eta_{BC}}\rho_B^{\frac{1}{2}}\\
&=\rho_B^{\frac{1}{2}}(\rho_B^{-\frac{1}{2}}\rho_{AB}\rho_B^{-\frac{1}{2}})^{\frac{1}{2}}\rho_B^{-\frac{
1}{2}}\rho_{BC}\rho_B^{-\frac{1}{2}}(\rho_B^{-\frac{1}{2}}\rho_{AB}\rho_B^{-\frac{1}{2}})^{\frac{1}{2}}\rho_B^{\frac{1}{2}}\\
&=(\Phi_{B\to AB}\otimes \id_C)(\rho_{BC}),
    \end{split}
    \end{equation}
    where $\Phi_{B\to AB}$ is defined as 
\begin{align}
\Phi_{B\to AB}(X):=&\rho_B^{1/2}(\rho_B^{-1/2}\rho_{AB}\rho_B^{-1/2})^{1/2}\rho_B^{-1/2}X\rho_B^{-1/2}(\rho_B^{-1/2}\rho_{AB}\rho_B^{-1/2})^{1/2}\rho_B^{1/2} \\ =&
d_B\rho_B^{1/2}\eta_{AB}^{1/2}\rho_B^{-1/2}X\rho_B^{-1/2}\eta_{AB}^{1/2}\rho_B^{1/2} \, .
\end{align}
Notice that $\Phi_{B\to AB}$ is linear and completely positive but not trace preserving. This map has an interesting form when rewritten using the polar decomposition $\rho_{AB}^{1/2}\rho_B^{-1/2}=d_B^{1/2}W_{AB}\eta_{AB}^{1/2}$ where $W_{AB}$ is a unitary and $\eta_{AB}$ as in Eq.\ \eqref{definition_eta}. 
We then obtain 
\begin{equation}\label{eq:MapPhi}
    \Phi_{B\to AB}(X_{B})=\rho_{AB}^{1/2}W_{AB}\rho_B^{-1/2}X_{B}\rho_B^{-1/2}W_{AB}^*\rho_{AB}^{1/2},
\end{equation}
which looks similarly to the Petz recovery map, but now with the additional unitary matrix $W_{AB}$ in between. Moreover,   the following trace inequality shows that $\Phi_{B\to AB}$ can only increase the trace by a factor depending on the dimension of $A$: for $X_B \geq 0$,
\begin{equation}
    \tr_{AB}[ \Phi_{B\to AB}(X_{B})] \leq d_A \Vert \rho_B^{-1}\Vert_{\infty}\tr_B X_{B}.
\end{equation}
 \begin{corollary}\label{coro:MapPhi}
     \label{theo:recoveryBSDPI} Let $\rho_{ABC}\in \cS(\cH_{ABC})$. Then $\rho_{ABC}$ is a BS-QMC if and only if $(\Phi_{B\to AB}\otimes \id_C)(\rho_{BC})=\rho_{ABC}$.    
 \end{corollary}
\begin{proof}
    Follows from Eq.\ \eqref{eq:FromBtoPhi} and the fact that $(\mathcal{B}_{B \rightarrow AB}\otimes \id_C)(\rho_{BC})=\rho_{ABC}$ if and only if $\rho_{ABC}$ is a BS-QMC.
\end{proof}

\section{(BS-)quantum Markov chains and their approximate versions} \label{BS-QMC-and-QMC}
\subsection{BS-quantum Markov chains which are not quantum Markov chains} \label{sec:BS-QMC-not-QMC}
In the previous sections, we have discussed the fact that the set of points that saturate the DPI for the relative entropy is contained in that of the BS-entropy, but the converse is not true \cite{HiaiMosonyi-f-divergences-2017,Jencova2009}. This can be translated to the simplified tripartite case of the conditional mutual information and the analogous BS quantities. In this case, we say that every QMC is a BS-QMC, but there are BS-QMCs which are not QMCs. An example of this is presented below.

\begin{example}\label{example:BSqmcNotQmc}
Consider a system with Hilbert spaces $\mathcal{H}_A=\mathcal{H}_B=\mathcal{H}_C=\mathbb{C}^2$ and let
\begin{equation}
    \rho_{ABC}=\frac{9}{47} \begin{pmatrix}
       \frac{1}{3} & 0 &0 &0 &0&0&0&0\\
       0 &\frac{4}{3}&0&-\frac{2}{3}&0&0&0&0\\
       0 & 0 & \frac{2}{3} &0&0&0&0&0\\
       0 & -\frac{2}{3}& 0& \frac{4}{3}&0&0&0&0\\
       0&0&0&0&\frac{1}{9}&0&\frac{1}{9}&0\\
       0&0&0&0&0&\frac{1}{3}&0&0\\
       0&0&0&0&\frac{1}{9}&0&\frac{4}{9}&0\\
       0&0&0&0&0&0&0&\frac{2}{3}\\
    \end{pmatrix}
\end{equation}
with marginals
\begin{equation}
    \rho_{BC}=\frac{9}{47}\begin{pmatrix}
        \frac{4}{9}&0 &\frac{1}{9}&0\\
        0&\frac{5}{3}&0&-\frac{2}{3}\\
        \frac{1}{9}&0&\frac{10}{9}&0\\
        0&-\frac{2}{3}&0&2
    \end{pmatrix},\; \rho_{AB}=\frac{9}{47}\begin{pmatrix}
        \frac{5}{3}& -\frac{2}{3}&0 &0\\
        -\frac{2}{3}&2&0&0\\
        0 &0 &\frac{4}{9}&\frac{1}{9}\\
        0&0&\frac{1}{9}&\frac{10}{9}
    \end{pmatrix}, \; \rho_B=\frac{9}{47}\begin{pmatrix}
        \frac{19}{9}&-\frac{5}{9}\\
        -\frac{5}{9}&\frac{28}{9}
    \end{pmatrix}.
\end{equation}
It can be checked that this state $\rho_{ABC}$ satisfies the BS recovery condition
\begin{equation}
    \rho_{ABC}=\rho_{AB} \rho_B^{-1} \rho_{BC},
\end{equation}
but that it is not a QMC since
\begin{equation}
    \rho_{ABC}\neq \rho_{AB}^{1/2}\rho_B^{-1/2}\rho_{BC}\rho_B^{-1/2}\rho_{AB}^{1/2}.
\end{equation}
    
\end{example}

A natural question is then what makes QMCs special in the set of BS-QMCs. The following result provides  different characterisations for the case of a BS-QMC to be a QMC.

\begin{proposition}\label{prop:QMC_in_BS}
    Let $\rho_{ABC}$ be a BS-QMC and let $\eta_{ABC}$ be the corresponding QMC. Let $\rho_{AB}^{1/2}\rho_B^{-1/2}=d_B^{1/2}W_{AB}\eta_{AB}^{1/2}$ be the polar decomposition, with $W_{AB}$ unitary. Then, the following are equivalent.
    \begin{enumerate}
        \item [(i)] $\rho_{ABC}$ is a QMC.
        \item [(ii)] There is a decomposition as in Theorem \ref{theo:StructureBSDPI} (v), such that also 
\begin{equation}\label{eq:rhoBStructure}
\rho_B=U_B^*\left(\bigoplus_n p_n\Tilde{\rho}_{B_n^L}\otimes \Tilde{\rho}_{B_n^R}\right)U_B
\end{equation}
for some $\Tilde{\rho}_{B_n^L}\in \cS(\cH_{B_n^L})$, $\Tilde{\rho}_{B_n^R}\in \cS(\cH_{B_n^R})$ and a probability distribution $\{p_n\}$.
\item [(iii)]  $[\rho_B^{it}\eta_{AB}\rho_B^{-it},\eta_{BC}]=0$ for all $t\in \mathbb R$.
\item [(iv)]  $W_{AB}\eta_{BC}W^*_{AB}=\eta_{BC}$.
\item[(v)] $d_B^{-1}\rho_{AB}^{-1/2}\rho_{ABC}\rho_{AB}^{-1/2}=\eta_{BC}$.
    \end{enumerate}
\end{proposition}

\begin{proof} 
\noindent \underline{$\text{(i)} \Rightarrow \text{(ii)}.$} It is easily seen from the structure theorem for QMC in \cite{HaydenJozsaPetzWinter-StrongSubadditivity-2004}. 

\vspace{0.1cm}

\noindent \underline{$\text{(ii)} \Rightarrow \text{(i)} .$} Suppose that $\rho_{ABC}$ admits a decomposition like \eqref{AnnaStates}, i.e.
\begin{equation}\label{eq:rhoABC_aux}
\rho_{ABC}=\rho_B^{1/2}U^*_B\left(\bigoplus_n d_B\, p_n \, \Tilde{\eta}_{AB_{n}^L}\otimes
\Tilde{\eta}_{B_n^RC}\right)U_B\rho_B^{1/2} \, ,
\end{equation}
for some unitary $U_B:\cH_B \to \oplus_n \cH_{B_n^L}\otimes \cH_{B_n^R}$  and let us assume also that $\rho_B$ has the form \eqref{eq:rhoBStructure} where $U_B$ is the same as in $\eqref{AnnaStates}$. If we trace out $A$ and $C$ in \eqref{eq:rhoABC_aux},
\begin{equation}
U_B^*\left(\bigoplus_n p_n \Tilde{\eta}_{B_n^L} \otimes   \Tilde{\eta}_{B_n^R}\right)U_B=\eta_B=\frac{1}{d_B}I_B,
\end{equation}
so for every sector $n$,
\begin{equation}
d_B p_n\Tilde{\eta}_{B_n^L}\otimes \Tilde{\eta}_{B_n^R}=I_{B_n^L}\otimes I_{B_n^R}\, ,
\end{equation} 
which implies that $d_Bp_n=d_{B_n^L}d_{B_n^R}$. Define now the quantum states
\begin{equation}
\xi_{AB_n^L}=d_{B_n^L}\Tilde{\rho}_{B_n^L}^{1/2}\Tilde{\eta}_{AB_{n}^L} \Tilde{\rho}_{B_n^L}^{1/2}, \quad \text{and} \quad \xi_{B_n^RC}=d_{B_n^R}\Tilde{\rho}_{B_n^R}^{1/2}\Tilde{\eta}_{B_{n}^RC} \Tilde{\rho}_{B_n^R}^{1/2} \, .
\end{equation}
Then,
\begin{equation}
\begin{split}
\rho_{ABC}&=U_B^*\left( \bigoplus_n d_Bp_n^2 \hspace{3pt}(\Tilde{\rho}_{B_n^L}^{1/2}\Tilde{\eta}_{AB_{n}^L} \Tilde{\rho}_{B_n^L}^{1/2}) \otimes (\Tilde{\rho}_{B_n^R}^{1/2}\Tilde{\eta}_{B_{n}^RC} \Tilde{\rho}_{B_n^R}^{1/2})\right)U_B\\
&=U_B^*\left( \bigoplus_n \frac{d_Bp_n^2}{d_{B_n^L}d_{B_n^R}}\xi_{AB_n^L} \otimes \xi_{B_n^RC} \right)U_B\\
&=U_B^*\left( \bigoplus_n p_n \xi_{AB_n^L} \otimes \xi_{B_n^RC} \right)U_B,
\end{split}
 \end{equation}
and we conclude by the structure theorem for QMC in \cite{HaydenJozsaPetzWinter-StrongSubadditivity-2004} that $\rho_{ABC}$ is a QMC.

\vspace{0.1cm}

\noindent \underline{$\text{(ii)} \Rightarrow \text{(iii)} .$}
It is also clear that (ii) implies (iii).

\vspace{0.1cm}

\noindent \underline{$\text{(iii)} \Rightarrow \text{(iv)} .$}
Assume (iii) and let $\mathcal A\subseteq
\mathcal{B}(\cH_{AB})$ be the unital subalgebra generated by $\rho_B^{it}\eta_{AB}\rho_B^{-it}$,
$t\in \mathbb R$. Then $\rho_B^{it}\mathcal A\rho_B^{-it}\subseteq \mathcal A$ and
$I_A\otimes \eta_{BC}\in \mathcal A'\otimes \mathcal{B}(\cH_C)$. Since $\mathcal A$ is a finite-dimensional C$^\ast$-algebra, there is a unitary $V_{AB}: \cH_{AB}\to \bigoplus_n \cH_n^{L}\otimes \cH_n^R$ such that we have the decomposition
\[
\mathcal A=V_{AB}^*\left(\bigoplus_n \mathcal{B}(\cH^L_n)\otimes I_{\cH_n^R}\right)V_{AB}.
\]
We obtain
\[
\eta_{AB}=V_{AB}^*(\oplus_n \Tilde{\eta}_n^L\otimes I_{\cH_n^R})V_{AB},\qquad
 \rho_B=V_{AB}^*(\oplus_n \Tilde{\rho}_n^L\otimes
\Tilde{\rho}_n^R)V_{AB},
\]
with some positive $\Tilde{\eta}_n^L \in \mathcal \mathcal{B}(\cH_n^L)$, and some positive-definite $\Tilde{\rho}_n^L\in \cS(\cH_n^L)$, $\Tilde{\rho}_n^R\in
 \cS(\cH_n^R)$. The last equality follows by \cite[Thm. 11.27]{Petz2008}.
We get that
\[
\rho_{AB}=d_B^{-1}\rho_B^{1/2}\eta_{AB}\rho_B^{1/2}=d_B^{-1}V_{AB}^*\left(\bigoplus_n
(\Tilde{\rho}_n^L)^{1/2}\Tilde{\eta}_n^L(\Tilde{\rho}_n^L)^{1/2}\otimes \Tilde{\rho}_n^R \right)V_{AB},
\]
and therefore $\rho_{AB}^{1/2}\rho_B^{-1/2}\in \mathcal A$, such that $W_{AB}\in \mathcal A$ by properties of the polar decomposition. Since $I_A\otimes \eta_{BC}\in \mathcal A'\otimes \mathcal{B}(\mathcal H_C)$, this implies (iv). 

\vspace{0.1cm}

\noindent \underline{$\text{(iv)} \Leftrightarrow \text{(v)} .$} Note that since $\rho_{ABC}$ is a BS-QMC, we have by  \cref{theo:StructureBSDPI} (iii) that
\[
\rho_{AB}^{-1/2}\rho_{ABC}\rho_{AB}^{-1/2}=d_B^2\rho_{AB}^{-1/2}\rho_B^{1/2}\eta_{AB}\eta_{BC}\rho_B^{1/2}\rho_{AB}^{-1/2}= d_BW_{AB}\eta_{BC}W^*_{AB},
\]
where we have used the polar decomposition to obtain $d_B^{1/2}\rho_{AB}^{-1/2}\rho_B^{1/2}=W_{AB}\eta_{AB}^{-1/2}$. It is now clear that (iv) is equivalent to (v). 

\vspace{0.1cm}

\noindent \underline{$\text{(iv)} \Rightarrow \text{(i)} .$} Finally, we have 
\[
\rho_{ABC}\hspace{-2pt}=\hspace{-2pt}(\Phi_{B\to AB}\otimes \id_C)(\rho_{BC})\hspace{-2pt}=\hspace{-2pt}d_B\rho_{AB}^{1/2}W_{AB}\eta_{BC}W_{AB}^*\rho_{AB}^{1/2}=d_B\rho_{AB}^{1/2}\eta_{BC}\rho_{AB}^{1/2}\hspace{-2pt}=\hspace{-2pt}\rho_{AB}^{1/2}\rho_B^{-1/2}\hspace{-3pt}\rho_{BC}\rho_B^{-1/2} \hspace{-3pt}\rho_{AB}^{1/2} \, ,
\]
which shows that (iv) implies (i) by \eqref{eq:PetzCondition}.
\end{proof}

Another important difference between QMCs and BS-QMCs is that while for QMCs $\rho_{AC}$ is always separable, for BS-QMCs this might be entangled. The next result provides us with examples of BS-QMCs such that $\rho_{AC}$ is NPT (non-positive partial transpose) entangled.

\begin{proposition}\label{prop:EntangledBSQMC}
Let  $\cH_B=\cH_{B_L} \otimes \cH_{B_R}$,  and assume $\cH_A=\cH_{B_L}=\cH_{B_R}=\cH_C=\mathbb{C}^d$. Consider
\begin{equation}
    \eta_{AB_L}=\frac{1}{d^2+d\alpha }(I_{AB_L}+\alpha F_{AB_L}) \, , \quad \eta_{B_RC}=\frac{1}{d^2+d\alpha }(I_{B_RC}+\alpha F_{B_RC}) \, ,
\end{equation}
Werner states with $\alpha \in [-1,1]$ and where $F_{XY}$ denotes the swap operator between the systems $X$ and $Y$. Consider $\vert \Phi \rangle=\frac{1}{\sqrt{2}}(\vert u_1 \otimes u_1 \rangle+\vert u_2 \otimes u_2 \rangle) \in \cH_B$, where  $u_1, u_2 \in \mathbb{R}^d$ are orthogonal and normalized.  If we define $X_B=I_A \otimes \vert \Phi \rangle \langle \Phi \vert \otimes I_C $, then $\rho_{ABC}=d^2X_B(\eta_{AB_L}\otimes\eta_{B_RC})X_B$ is a BS-QMC and  $\rho_{AC}$ has negative partial transpose  if and only if $\alpha \in [-1, 1-\sqrt{3})$. In particular, for $\alpha < 1-\sqrt{3}$, $\rho_{AC}$ is entangled, implying that $\rho_{ABC}$ is not a QMC.
\end{proposition}
\begin{proof}
    Let $\alpha \in [-1,1]$ and consider the tensor product of Werner states $\eta_{ABC}=\eta_{AB_L}\otimes \eta_{B_RC}$, which is a QMC by \eqref{eq:StructuralQMC} and further satisfies
\begin{equation}
\eta_B=\tr_{AC} [\eta_{ABC}]=\tr_{A}[\eta_{AB_L}]\tr_{C}[\eta_{B_RC}]=\frac{1}{d^2}I_B.
\end{equation}
Consider now $\vert \Phi \rangle=\frac{1}{\sqrt{2}}(\vert u_1 \otimes u_1 \rangle+\vert u_2 \otimes u_2 \rangle) \in \cH_B$ and define $X_B=I_A \otimes \vert \Phi \rangle \langle \Phi \vert \otimes I_C $, which satisfies $X_B^{1/2}=X_B$. Putting it all together we can define $\rho_{ABC}=d^2 X_B\eta_{ABC}X_B$ which is a  BS-QMC by \cref{remark:ConverseEta}. Now, if we write 
\begin{equation}
    \vert \Phi \rangle \langle \Phi \vert =\frac{1}{2}\sum_{i,j \in \{1,2\}} \vert u_i \rangle \langle u_j \vert \otimes  \vert u_i \rangle \langle u_j \vert
\end{equation}
we can decompose \begin{equation}\rho_{ABC}=\frac{d^2}{4(d^2+\alpha d)^2}\sum_{i,j,k,l \in \{1,2\} }\rho_{AB_L}^{(ijkl)}\otimes \rho_{B_RC}^{(ijkl)} \, 
 ,\end{equation} with
\begin{equation}
    \rho_{AB_L}^{(ijkl)}=\left(I_{A} \otimes \vert u_i \rangle \langle u_j\vert\right)\left( I_{AB_L}+\alpha F_{AB_L}\right) \left(I_{A} \otimes \vert u_k \rangle \langle u_l\vert\right) \, ,
\end{equation}
and
\begin{equation}
    \rho_{B_RC}^{(ijkl)}=\left( \vert u_i \rangle \langle u_j\vert \otimes I_C \right)\left( I_{B_RC}+\alpha F_{B_RC}\right) \left( \vert u_k \rangle \langle u_l\vert \otimes I_C\right) \, .
\end{equation}
Consequently,
\begin{equation}
    \rho_{AC}=\frac{1}{4(d+\alpha )^2}\sum_{i,j,k,l \in \{1,2\} }\tr_{B_L}[\rho_{AB_L}^{(ijkl)}]\otimes \tr_{B_R}[\rho_{B_RC}^{(ijkl)}].
\end{equation}
From the partial trace with respect to the system $\cH_{B_L}$ we obtain the marginal
\begin{equation}
\begin{split}
    \rho_A^{(ijkl)}&=\tr_{B_L}[\rho_{AB_L}^{(ijkl)}]\\
    &=\tr_{B_L}[I_A\otimes \vert u_i  \rangle \langle u_j \vert \vert u_k \rangle \langle u_l \vert] +\alpha\tr_{B_L}[ \underbrace{ \left(I_{A} \otimes \vert u_i \rangle \langle u_j\vert\right) F_{AB_L}\left(I_{A} \otimes \vert u_k \rangle \langle u_l\vert\right)}_{\vert u_k \otimes u_i \rangle \langle u_j \otimes u_l \vert }] \\
    &=\delta_{il}\left(\delta_{jk} I_A+\alpha \vert u_k \rangle \langle u_j \vert \right) 
\end{split}
\end{equation}
and analogously $\rho_C^{(ijkl)}=\delta_{il}\left(\delta_{jk} I_C+\alpha \vert u_k \rangle \langle u_j \vert \right) $. The marginal $\rho_{AC}$ can now be expressed as
\begin{equation}
    \rho_{AC}=\frac{2}{4(d+\alpha )^2}\left[2I_{AC} + \sum_{j \in \{1,2\}}\alpha \left( I_A\otimes \vert u_j \rangle \langle u_j \vert+\vert u_j \rangle \langle u_j \vert \otimes I_C \right)+\alpha^2\sum_{j,k \in \{1,2\}}  \vert u_k \rangle \langle u_j \vert \otimes \vert u_k \rangle \langle u_j \vert \right] \, .
\end{equation}

Now, the partial transpose over the first system $T_A$ is a map of the form $T\otimes I_C$, where $T$ is the usual transposition. With this 
\begin{equation}
    \rho_{AC}^{T_A}=\frac{1}{2(d+\alpha )^2}\left[2I_{AC} + \sum_{j \in \{1,2\}}\alpha \left( I_A\otimes \vert u_j \rangle \langle u_j \vert+\vert u_j \rangle \langle u_j \vert \otimes I_C \right)+\alpha^2\sum_{j,k \in \{1,2\}}  \vert u_j \rangle \langle u_k \vert \otimes \vert u_k \rangle \langle u_j \vert \right] \, .
\end{equation}
To check that $\rho_{AC}$ has negative partial transpose, we need to find a vector $v\in \cH_A \otimes \cH_C$ such that $\langle v, \rho_{AC}^{T_A} \, v \rangle<0$. The most general vector $v$ that contributes with the linear and quadratic terms on $\alpha$ is of the form
\begin{equation}
    v=\gamma_1 u_1 \otimes u_1+\gamma_2 u_1 \otimes u_2+ \gamma_3 u_2 \otimes u_1+ \gamma_4 u_2 \otimes u_2
\end{equation}
with $\gamma=(\gamma_1,\gamma_2,\gamma_3,\gamma_4) \in \mathbb{C}^4$ satisfying $\Vert \gamma \Vert_2^2=1$, i.e. we assume that $v$ is normalized. Now we compute the expectation for the vector $v$ of each one of the terms of $\rho_{AC}^{T_A}$ and obtain:
\begin{subequations}
\begin{equation}
 \langle v,I_{AC} \, v \rangle=\Vert v \Vert_2^2=1 \, ,
\end{equation}
\begin{equation}
\sum_{j \in \{1,2\}} \langle v, (I_A\otimes \vert u_j \rangle \langle u_j \vert) v \rangle=\Vert v \Vert_2^2 \, ,
\end{equation}
\begin{equation}
\sum_{j \in \{1,2\}} \langle v, ( \vert u_j \rangle \langle u_j \vert \otimes I_C) v \rangle=\Vert v \Vert_2^2 \, ,
\end{equation}
\begin{equation}
\sum_{j,k \in \{1,2\}} \langle v, (\vert u_j \rangle \langle u_k \vert \otimes \vert u_k \rangle \langle u_j \vert) v \rangle=\vert\gamma_1\vert^2+\vert \gamma_4\vert^2+2\text{Re}[\overline{\gamma_2}\gamma_3] \, .
\end{equation}
\end{subequations}
Plugging in these values, we can now write
\begin{equation}
2(d+\alpha)^2\langle v, \rho_{AC}^{T_A} \, v\rangle= 2(1+\alpha)+\alpha^2(\vert\gamma_1\vert^2+\vert\gamma_4\vert^2+2\text{Re}[\overline{\gamma_2}\gamma_3]) \, .
\end{equation}
The minimum of the right-hand side in terms of $\gamma=(\gamma_1,\gamma_2,\gamma_3,\gamma_4)$ is clearly achieved whenever $\gamma_1=\gamma_4=0$ and Re$[\overline{\gamma_2}\gamma_3]=-1/2$ i.e. for the values $(\gamma_2,\gamma_3)=(\pm\frac{1}{\sqrt{2}},\mp\frac{1}{\sqrt{2}})$ or $(\gamma_2,\gamma_3)=(\pm\frac{i}{\sqrt{2}},\mp\frac{i}{\sqrt{2}})$,  giving both cases the same result for the expectation value:
\begin{equation}
    2(d+\alpha)^2\langle v, \rho_{AC}^{T_A} \, v\rangle=-\alpha^2+2\alpha+2=-(\alpha-(1-\sqrt{3}))(\alpha-(1+\sqrt{3})),
\end{equation}
which is then negative for the values $\alpha \in [-1,1-\sqrt{3})$. The result follows then by the fact that negative partial transpose implies entanglement.
\end{proof}
This last result provides BS-QMCs which have marginal in $AC$ entangled and therefore cannot be  QMCs.  However, not every BS-QMC $\rho_{ABC}$ that is not a QMC satisfies that $\rho_{AC}$ is entangled as illustrated by \cref{example:BSqmcNotQmc}, where it is separable. 
These examples therefore illustrate the fundamental difference of the BS-CMI being zero and the CMI beging zero, because the set where the  BS-CMI vanishes contains also states where $\rho_{AC}$ is NPT entangled, whereas states for which the CMI vanishes always have $\rho_{AC}$  separable.

\begin{remark}
Another remarkable feature of the correspondence between BS-QMCs and QMCs provided by the identification by $\eta$ arises from \Cref{prop:EntangledBSQMC}. Note that we have constructed an example of a family of states $\rho_{ABC}$ which are BS-QMCs and for which $\rho_{AC}$ is entangled. However, if we define now the corresponding $\eta_{ABC}$ associated to $\rho_{ABC}$, it is a QMC by \Cref{theo:StructureBSDPI}, and therefore $\eta_{AC}$ is separable. Therefore, the construction provided by $\eta$ is entanglement-breaking  between $A$ and $C$  when restricted to set of BS-QMCs. Conversely, the correspondence given by $\omega(X)$  in Remark \ref{remark:ConverseEta} can create entanglement between $A$ and $C$ when restricted to the set of QMCs for some quantum states $X_B$. 
\end{remark}

\subsection{Approximate BS-quantum Markov chains and quantum Markov chains}

 Theorem \ref{theo:StructureBSDPI} provides  an exact identification between QMCs and BS-QMCs. A natural question is then whether an equivalence between approximate versions of these notions holds as well. We say that $\rho_{ABC}$ is an \textit{$\varepsilon$-approximate QMC} if
\begin{equation}
    I_\rho(A:C|B) \leq \varepsilon \, ,
\end{equation}
and analogously $\rho_{ABC}$ is an \textit{$\varepsilon$-approximate BS-QMC} if 
\begin{equation}
    \widehat{I}^{\mathrm{rev}}_\rho(A;C|B) \leq \varepsilon \, . 
\end{equation}

To explore the connection between approximate QMCs and approximate BS-QMCs, we first provide a general lower bound for the reversed BS-CMI of $\rho_{ABC}$ in terms of the CMI of $\eta_{ABC}$, and conversely under some constraints.
\begin{proposition}\label{prop:relation_Irev_I}
 There exists a positive non-zero function $f$ such that for any finite-dimensional $\mathcal{H}_{ABC}$ and any invertible $\rho_{ABC} \in \cS(\cH_{ABC})$ with associated  $\eta_{ABC}$, we have 
\begin{equation}
    \widehat{I}^{\mathrm{rev}}_{\rho}(A;C | B)\geq f(\rho,d_A,d_B, d_C)I_{\eta}(A:C\vert B)^8 \, .
\end{equation}
Conversely, there exist  positive  non-zero functions $g$ and $h$ such that, in the conditions above, if $[\eta_{AB},\eta_{BC}]=0$, then 
\begin{equation}\label{ineq:BoundReversedCMI1}
    \widehat{I}^{\mathrm{rev}}_{\rho}(A;C | B)\leq g(\rho,d_A,d_B, d_C)I_{\eta}(A:C\vert B)^{1/4}\, ,
\end{equation}
and 
\begin{equation}\label{ineq:BoundReversedCMI2}
    \widehat{I}^{\mathrm{rev}}_{\rho}(A;C | B)\leq h(\rho,d_A,d_B, d_C)I_{\eta}(A:C\vert B)^{1/2}\, .
\end{equation}

\end{proposition}
\begin{proof}
To simplify notation hereafter, let us denote $\mathcal{T}(\sigma_{AB}) \equiv (\id_A \otimes \mathcal{T}_{B \rightarrow BC} )(\sigma_{AB})$ for $\mathcal{T} \in \{ \mathcal{B}, \mathcal{P}, \Phi\}$ and $\sigma_{ABC}\in \{ \rho_{ABC}, \eta_{ABC}\}$ depending on each case. Note that $\mathcal{B}, \mathcal{P}, \Phi$ are defined with respect to different states, those on whose marginal they are evaluated throughout the proof.
The trace norm between $\eta_{ABC}$ and its recovery channel can be expressed in terms of $\rho_{ABC}$ and $\Phi$ as follows:
    \begin{equation}
            \Vert \eta_{ABC}-\mathcal{P}(\eta_{AB})\Vert_1 
             = \frac{1}{d_B}\Vert \rho_B^{-1/2}\rho_{ABC}\rho_B^{-1/2} -\rho_B^{-1/2}\Phi(\rho_{AB})\rho_B^{-1/2}\Vert_1 \, .
    \end{equation}
    Consider the completely positive and trace non-increasing multiplication map \begin{equation}
    \mathcal{M}_{\rho_B}(X):=\frac{1}{\Vert \rho_B^{-1}\Vert_{\infty}}\rho_B^{-1/2}X \rho_B^{-1/2},
    \end{equation}
    for $X\geq 0$. $\mathcal{M}_{\rho_B}$ satisfies then the DPI for the trace-distance and as a consequence,
\begin{equation}
        \Vert \rho_{ABC}-\Phi(\rho_{AB})\Vert_1\geq  \frac{d_B}{\Vert \rho_B^{-1}\Vert_{\infty}}\Vert \eta_{ABC}-\mathcal{P}(\eta_{AB})\Vert_1.
    \end{equation}
    To conclude this part, we make use now of 
    \begin{equation}
        I_{\eta}(A:C\vert B)\leq 2 (\log\min\{d_A,d_C\}+1)\Vert \eta_{ABC}-\mathcal{P}(\eta_{AB})\Vert_1^{1/2} \, ,
    \end{equation}
    which can be found in \cite{bluhm2023general}, and together with Eq.\ \eqref{ineq:LowerBoundReversedCMI}, we  conclude
    \begin{equation}
    \widehat{I}^{\mathrm{rev}}_{\rho}(A;C | B)\geq  2^{-8}(\log\min\{d_A,d_C\}+1)^{-8}\left( \frac{ d_B \pi}{8\Vert \rho_B^{-1}\Vert_{\infty}} \right)^4 \|\rho_{BC}^{-1/2}\rho_{ABC}\rho_{BC}^{-1/2} \|_\infty^{-2}I_{\eta}(A:C\vert B)^8 \, ,
\end{equation}
   which is well-defined as $\rho_{BC}^{-1/2}\rho_{ABC}\rho_{BC}^{-1/2}\neq 0$, since $\tr(\rho_{BC}^{-1/2}\rho_{ABC}\rho_{BC}^{-1/2})=d_Bd_C \neq 0$, and also non-zero.

We prove now the converse. Firstly, using part of the proof of \cite[Theorem 3.6]{gondolf2024conditional}
\begin{equation}
    \begin{split}
        &\widehat{I}^{\mathrm{rev}}_{\rho}(A;C | B) \\ & \hspace{1cm}  \leq \Vert \rho_{BC}^{-1/2}\rho_{ABC}\rho_{BC}^{-1/2}\Vert_{\infty}(\Vert \rho_B^{-1}\Vert_{\infty} \Vert \rho_B \Vert_{\infty})^{1/2}\Vert \rho_{ABC}^{-1}\rho_{BC}\Vert_{\infty}\Vert \rho_{ABC}\rho_{AB}^{-1}\rho_B \rho_{BC}^{-1}- I \Vert_{\infty}\\
        &\hspace{1cm} =\Vert \rho_{BC}^{-1/2}\rho_{ABC}\rho_{BC}^{-1/2}\Vert_{\infty}(\Vert \rho_B^{-1}\Vert_{\infty} \Vert \rho_B \Vert_{\infty})^{1/2}\Vert \rho_{ABC}^{-1}\rho_{BC}\Vert_{\infty}\Vert \rho_{ABC}\, \mathcal{B}(\rho_{AB})^{-1}- I \Vert_{\infty}\\
         & \hspace{1cm}  \leq  \Vert \rho_{BC}^{-1/2}\rho_{ABC}\rho_{BC}^{-1/2}\Vert_{\infty}(\Vert \rho_B^{-1}\Vert_{\infty} \Vert \rho_B \Vert_{\infty})^{1/2}\Vert \rho_{ABC}^{-1}\rho_{BC}\Vert_{\infty} \Vert \mathcal{B}(\rho_{AB})^{-1} \Vert_{\infty}   \Vert \rho_{ABC}-\mathcal{B}(\rho_{AB})\Vert_{\infty}\\
         & \hspace{1cm}  \leq \Vert \rho_{BC}^{-1/2}\rho_{ABC}\rho_{BC}^{-1/2}\Vert_{\infty}(\Vert \rho_B^{-1}\Vert_{\infty} \Vert \rho_B \Vert_{\infty})^{1/2}\Vert \rho_{ABC}^{-1}\rho_{BC}\Vert_{\infty} \Vert \mathcal{B}(\rho_{BC})^{-1} \Vert_{\infty}   \Vert \rho_{ABC}-\mathcal{B}(\rho_{AB})\Vert_{1}\\& \hspace{1cm}=: g_1(\rho,d_A,d_B,d_C)\Vert \rho_{ABC}-\mathcal{B}(\rho_{AB})\Vert_{1}\, . 
    \end{split}
\end{equation}
Consider now the operator norm for the channel
    \begin{equation}
        \Vert \mathcal{M}_{\rho_B}\Vert=\sup_{\Vert Y \Vert_2 \leq 1}\Vert \mathcal{M}_{\rho_B} Y\Vert_2,
    \end{equation}
    where $\Vert \hspace{2pt} \cdot \hspace{2pt}\Vert_2$ denotes the Frobenius norm. Since $\mathcal{M}_{\rho_B}$ is invertible, 
    \begin{equation}
        \Vert \mathcal{M}_{\rho_B} Y\Vert_1\geq \Vert \mathcal{M}_{\rho_B} Y\Vert_2\geq \Vert \mathcal{M}_{\rho_B}^{-1}\Vert^{-1} \Vert Y \Vert_2\geq     \frac{1}{\sqrt{d_{ABC}}}\Vert \mathcal{M}_{\rho_B}^{-1}\Vert^{-1} \Vert Y \Vert_1  \, .
    \end{equation} 
By letting $Y=\rho_{ABC}-\Phi(\rho_{AB})$,
\begin{equation}
    \Vert \eta_{ABC}-\mathcal{P}(\eta_{AB})\Vert_1 \geq \frac{1}{d_B\sqrt{d_{ABC}}} \Vert \rho_{ABC}-\Phi(\rho_{AB})\Vert_1  \, ,
\end{equation}
since $\|\mathcal M^{-1}_{\rho_B}\| \leq \|\rho_B^{-1}\|_\infty$.

 To obtain inequality Eq.\ \eqref{ineq:BoundReversedCMI1}, we use Eq.\  \eqref{eq:carlen-vershynina} from \cite{CarlenVershynina-Stability-DPI-RE-2017} and obtain
\begin{equation}
    \begin{split}
        I_{\eta}(A:C\vert B)&\geq \left(\frac{\pi}{8}\right)^4\Vert \eta_{ABC}^{-1}\Vert_{\infty}^{-2} \Vert (\tau_A \otimes \eta_B \otimes \tau_C)^{-1}\Vert_{\infty}^{-2}\Vert \eta_{ABC}-\mathcal{P}(\eta_{AB})\Vert_1^4\\
        &\geq \left(\frac{\pi}{8}\right)^4\Vert \eta_{ABC}^{-1}\Vert_{\infty}^{-2}\left( \frac{1}{d_A d_B d_C}
 \right)^2\left(\frac{1}{ d_B\sqrt{d_Ad_Bd_C}}\right)^4\Vert  \rho_{ABC}-\Phi(\rho_{AB})\Vert_1^4\\
        &=\left(\frac{\pi}{8}\right)^4\Vert \eta_{ABC}^{-1}\Vert_{\infty}^{-2}\left(\frac{1}{d_Ad_B^{2}d_C} \right)^4\Vert  \rho_{ABC}-\mathcal{B}(\rho_{AB})\Vert_1^4 \\
        &=: g_2(\rho, d_A,d_B,d_C)\Vert  \rho_{ABC}-\mathcal{B}(\rho_{AB})\Vert_1^4 \, ,
    \end{split}
\end{equation}
since $[\eta_{AB},\eta_{BC}]=0$ (see Eq.\ \eqref{eq:FromBtoPhi}).
Combining the bounds for the CMI,
\begin{equation}
  \widehat{I}^{\mathrm{rev}}_{\rho}(A;C | B)\leq g_1(\rho,d_A,d_B,d_C)g_2(\rho,d_A,d_B,d_C)^{-1/4}I_\eta(A:C | B)^{1/4}
\end{equation}
and Eq.\ \eqref{ineq:BoundReversedCMI1} holds by letting $g=g_2^{-1/4}g_1$, which is again well defined and non-zero following the same arguments as before.

For inequality Eq.\  \eqref{ineq:BoundReversedCMI2}, we consider the following rotated Petz recovery map for $\eta_{ABC}$,
\begin{subequations}
\begin{align}
    \mathcal{R}^t_{B \to BC}(X_B)&:=\eta_{BC}^{-it} \, \mathcal{P}_{B \to BC}(\eta_B^{it} \, X_B \, \eta_B^{-it})\, \eta_{BC}^{it}\\
    &=\eta_{BC}^{\frac{1}{2}-it}\eta_B^{-\frac{1}{2}+it}X_B \eta_B^{-\frac{1}{2}-it}\eta_{BC}^{\frac{1}{2}+it}\\
    &=d_B \, \eta_{BC}^{\frac{1}{2}-it} X_B\eta_{BC}^{\frac{1}{2}+it}\, ,
\end{align}
\end{subequations}
since $\eta_B=\tau_B$. Applying the map  on $\eta_{AB}$ we obtain 
\begin{equation}
    (\id_A \otimes \mathcal{R}^t_{B \to BC} )(\eta_{AB})=d_B \, \eta_{BC} \, \eta_{AB}=(\id_A \otimes \mathcal{P}_{B \to BC} )(\eta_{AB})\, ,
\end{equation}
which is independent of $t$, since $[\eta_{AB},\eta_{BC}]=0$. Consequently, using \cite[Theorem 4]{wilde-2015}, the inequality $\log(x)\leq x-1$ for $x>0$, the Fuchs-van de Graaf inequality and the previous arguments
\begin{subequations}
\begin{align}
    I_{\eta}(A:C\vert B)&\geq -\log \left(\sup_{t \in \mathbb R}F(\eta_{ABC},\mathcal{R}^t(\eta_{AB}))\right)\\
    &= -\log F(\eta_{ABC},\mathcal{P}(\eta_{AB}))\\
    &\geq  1- F(\eta_{ABC},\mathcal{P}(\eta_{AB}))\\
    & \geq \frac{1}{4}\Vert \eta_{ABC}-\mathcal{P}(\eta_{AB})\Vert_1^2\\
    &\geq \frac{1}{4d_B^2 d_{ABC}} \Vert \rho_{ABC}-\mathcal{B}(\rho_{AB})\Vert_1^2 \, .
\end{align}
\end{subequations}
The result follows then by letting $h=2d_B\sqrt{d_{ABC}}g_1$.
 \end{proof}

Consequently, if $\rho_{ABC}$ is an invertible approximate BS-QMC, then $\eta_{ABC}$ is an approximate QMC. For the second part of \Cref{prop:relation_Irev_I}, notice that, even though we are assuming the constraints $[\eta_{AB},\eta_{BC}]=0$ and $\eta_B=\tau_B$, this does not necessarily imply  that $\eta_{ABC}$ is a QMC. Moreover, by taking the quotient of the right-hand side of inequalities Eq.\ \eqref{ineq:BoundReversedCMI1} and Eq.\ \eqref{ineq:BoundReversedCMI2}, we obtain that the first bound is tighter than the second whenever 
\begin{equation}
    I_{\eta}(A:C\vert B)\geq \left(\frac{4}{\pi}\right)^4(d_{ABC})^2\Vert \eta_{ABC}^{-1}\Vert_{\infty}^2 \, .
\end{equation}

 With the aim of providing the second part of \Cref{prop:relation_Irev_I} without the further assumption $[\eta_{AB}, \eta_{BC}]=0$, a first step can be an upper bound for the reversed CMI of $\rho_{ABC}$ in terms of the distance from $\rho_{ABC}$ to one of its recoveries. In \cite[Theorems 3.2 and 3.6]{gondolf2024conditional}, such a result is proven in terms of $\mathcal{B}(\rho_{BC})$. It is desirable though to construct another bound as a distance from $\rho_{ABC}$ to a symmetric recovery. For that, we need to introduce the rotated version of $\Phi_{B\to AB}$, namely
\begin{equation}
    \Phi_{B\to AB}^{\mathrm{rot}} (X) := \int_{-\infty}^{+\infty} dt \beta_0(t) \rho_B^{\frac{1-it}{2}}(\rho_B^{-1/2}\rho_{AB}\rho_B^{-1/2})^{\frac{1-it}{2}}\rho_B^{\frac{-1+it}{2}}X \, \rho_B^{\frac{-1-it}{2}}(\rho_B^{-1/2}\rho_{AB}\rho_B^{-1/2})^{\frac{1+it}{2}}\rho_B^{\frac{1+it}{2}} \, , 
\end{equation}
with $\beta_0(t)= \frac{\pi}{2( \cosh (\pi t) +1)}$. We leave as an open question whether, whenever $\rho_{ABC}$ is an exact BS-QMC, 
\begin{equation}\label{eq:fixed_points_Phirot}
    (\Phi_{B\to AB}^{\mathrm{rot}}\otimes \id_C) (\rho_{BC}) = \rho_{ABC} \, ,
\end{equation}
and thus
$\Phi_{B\to AB}^{\mathrm{rot}}$ is another recovery condition for BS-QMCs. The following result is an immediate consequence of the multivariate trace inequalities of Sutter et al. \cite{Sutter2017b}, and in particular shows that whenever Eq.\ \eqref{eq:fixed_points_Phirot} holds, $\rho_{ABC}$ is a BS-QMC.

\begin{proposition} \label{prop:appox-rotated2}
 Let $\rho_{ABC}$ be a positive-definite quantum state and let $\eta_{ABC}= \frac{1}{d_B} \rho_B^{-1/2} \rho_{ABC}\rho_B^{-1/2} $. Then,
 \begin{equation}
      \widehat{I}^{\mathrm{rev}}_{\rho}(A;C | B) \leq \frac{1}{d_C} \norm{\rho_{ABC}^{-1/2} \, \rho_{AB}^{1/2}}_{\infty}^2 \norm{   ( \id_A \otimes \Phi_{B\to BC}^{\mathrm{rot}} )(\rho_{AB})  -\rho_{ABC}}_1\, .
 \end{equation}
\end{proposition}

\begin{proof}

We drop the factors $\tau_C$ whenever they are unnecessary and the dimensional factors cancel out to simplify notation. We first rewrite $\widehat{I}^{\mathrm{rev}}_{\rho}(A;C | B) $ as 
\begin{align}
    \widehat{I}^{\mathrm{rev}}_{\rho}(A;C | B) & =  \widehat{D}(\rho_{AB} \otimes \tau_C  \| \rho_{ABC}) - \widehat{D}(\rho_{B}\otimes \tau_C  \| \rho_{BC}) \\
    &  = \tr \left[ \rho_{AB} \otimes \tau_C \left( \log (\rho_{AB}^{1/2}   \, \rho_{ABC}^{-1} \, \rho_{AB}^{1/2}   ) - \log (\rho_{B}^{1/2}   \, \rho_{BC}^{-1} \, \rho_{B}^{1/2}   ) \right. \right. \\
    & \hspace{3cm} \left. \left. -\log (\rho_{AB} \otimes \tau_C  ) + \log (\rho_{AB} \otimes \tau_C) - \log \rho_B + \log \rho_B    \right) \right] \\
    & = -  D(\rho_{AB} \otimes \tau_C \,  \|  \, \Omega) \, ,
\end{align}
where 
\begin{align}
    \Omega& := \exp \Big\{ \log (\rho_{AB}^{1/2} \, \rho_{ABC}^{-1} \, \rho_{AB}^{1/2} ) + \log (\rho_{B}^{-1/2}  \, \rho_{BC} \, \rho_{B}^{-1/2}  )   + \log (\rho_{AB} \otimes \tau_C ) - \log \rho_B + \log \rho_B     \Big\} \, .
\end{align}
Using the fact that the relative entropy between any two density matrices is always non-negative, we have
\begin{align}
 \widehat{I}^{\mathrm{rev}}_{\rho}(A;C | B)
    &  \leq   \log \tr[ \Omega ] \\
    &  \leq   \log \tr \Big[  \int_{-\infty}^{+\infty} dt \beta_0(t)  \left( \rho_{AB}^{1/2} \, \rho_{ABC}^{-1} \, \rho_{AB}^{1/2}  \right) \rho_B^{\frac{1-it}{2}} \left(\rho_B^{-1/2}\rho_{BC}\rho_B^{-1/2}\right)^{\frac{1-it}{2}}  \rho_B^{\frac{-1+it}{2}} \rho_{AB} \otimes \tau_C  \,    \\
    & \hspace{4cm} \rho_B^{\frac{-1-it}{2}} \left(\rho_B^{-1/2}\rho_{BC}\rho_B^{-1/2}\right)^{\frac{1+it}{2}}  \rho_B^{\frac{1+it}{2}}   \Big] \\
    &  =  \log \tr[  \left( \rho_{AB}^{1/2} \, \rho_{ABC}^{-1} \, \rho_{AB}^{1/2}  \right)( \id_A \otimes \Phi_{B\to BC}^{\mathrm{rot}}  )(\rho_{AB} \otimes \tau_C) -   \rho_{AB}\otimes \tau_C    +  \rho_{AB}\otimes \tau_C ] \\
    &  \leq  \; \tr[  \left( (\rho_{AB} \otimes \tau_C)^{1/2} \, \rho_{ABC}^{-1} \, (\rho_{AB} \otimes \tau_C)^{1/2}  \right) ( \id_A \otimes \Phi_{B\to BC}^{\mathrm{rot}} )(\rho_{AB}) -  \rho_{AB}\otimes \tau_C] \\
    &  \leq  \frac{1}{d_C} \norm{\rho_{ABC}^{-1/2} \, \rho_{AB}^{1/2}}_{\infty}^2 \norm{   ( \id_A \otimes \Phi_{B\to BC}^{\mathrm{rot}} )(\rho_{AB})  -\rho_{ABC}}_1 \, ,
\end{align}

\end{proof}

We conclude the section by summarizing the results presented here and leaving some open questions. We have explored a possible equivalence between $\rho_{ABC}$ being an approximate BS-QMC and $\eta_{ABC} $ being an approximate QMC, by constructing inequalities between $\widehat{I}_\rho^{\operatorname{rev}}(A;C|B)$ and $I_\eta(A:C |B)$. We have shown that the former is always lower bounded by a function of the latter, and a reverse bound holds under the additional assumption of commuting marginals of $\eta_{ABC}$.   

We have also explored the relation of $\widehat{I}_\rho^{\operatorname{rev}}(A;C|B)$ and the distance between $\rho$ and one of its BS-recovery conditions. A lower bound can be proven in terms of the distance to $\Phi(\rho_{BC})$, as shown in \cite{gondolf2024conditional}. We have shown an upper bound in terms of the distance to $\Phi^{\operatorname{rot}}(\rho_{BC})$, and left as an open question whether an upper bound can be found in terms of  $\Phi(\rho_{BC})$ directly, which would show in particular that the fixed points of $\Phi^{\operatorname{rot}}$ coincide with those of $\Phi$. In this case $\Phi^{\operatorname{rot}}$ would constitute another recovery condition for BS-QMCs.

\section{Equality conditions for DPIs of BS- and relative entropy}\label{sec:equality_conditions}

For a general quantum channel $\mathcal{T}$ and a state $\sigma$, the recovery condition $\mathcal{B}^\sigma_{\mathcal{T}}$  for the data-processing inequality of the BS-entropy was found in \cite{BluhmCapel-BSentropy-2019}. To be able to map quantum states into positive matrices, we present a  recovery condition which is non-linear but preserves positivity defined by
\begin{equation}
    \mathcal{B}^{\sigma,\text{sym}}_{\mathcal{T}}\circ \mathcal{T} (\rho):= ( \sigma \mathcal{T}^*(\mathcal{T}(\sigma)^{-1} \mathcal{T}(\rho)^2\mathcal{T}(\sigma)^{-1}) \sigma )^{1/2}  \, .
\end{equation}
However, the non-linearity of this recovery map can make this quantity a difficult  object to deal with. In Corollary \ref{coro:recoveryPhiChannels} we will see how to actually construct a completely positive, non trace-preserving  linear map, which will look very similar to the Petz recovery map.

The BS-entropy belongs to a larger family of entropies called \textit{maximal $f$-divergences} \cite{Matsumoto2018,HiaiMosonyi-f-divergences-2017}, defined as $\widehat{D}_f(\rho \| \sigma) =\tr[\sigma f(\sigma^{-1/2}\rho\sigma^{-1/2})]$   for any operator convex function $f$ on $[0,\infty)$. Since the saturation of the data-processing inequality is equivalent for every non-linear operator convex $f$, we can use the same recovery conditions for all of them.

\begin{theorem}\label{thm:equivalence_recovery_conditions}
    Let $\rho, \sigma$ be two quantum states, with $\sigma$ invertible, and let $\mathcal{T}$ be a quantum channel. The following are equivalent: 
    \begin{enumerate}
        \item[(i)] $\widehat{D}(\rho \| \sigma) = \widehat{D}(\mathcal{T}(\rho) \| \mathcal{T}(\sigma) ) $.
        \item[(ii)] $\widehat{D}_f(\rho \| \sigma) = \widehat{D}_f(\mathcal{T}(\rho) \| \mathcal{T}(\sigma) ) $, for every operator convex function $f$ on $[0,\infty)$.
        \item[(iii)] $\rho=\mathcal{B}^\sigma_{\mathcal{T}} \circ \mathcal{T} (\rho)$ .
        \item[(iv)] $\rho =\mathcal{B}^{\sigma,sym}_{\mathcal{T}} \circ \mathcal{T} (\rho)$ .
    \end{enumerate}
\end{theorem}

\begin{proof} 
   \noindent \underline{$\text{(i)} \Leftrightarrow \text{(ii)} \Leftrightarrow \text{(iii)}$.} The first  equivalence was proven in \cite[Theorem 3.34]{HiaiMosonyi-f-divergences-2017}, and the second in \cite{BluhmCapel-BSentropy-2019}.

\vspace{0.1cm}

\noindent \underline{$\text{(iii)} \Rightarrow \text{(iv)} .$} 
  Let us assume that (ii) holds. Then,   $ \rho = \sigma \mathcal{T}^\ast ( \mathcal{T}(\sigma)^{-1} \mathcal{T}(\rho) ) $, 
    and hence since $\rho^2=\rho\rho^*$,
    \begin{subequations}\label{eq:rho_square_equiv_recovery_conditions}
    \begin{align}
        \rho^2 &= \sigma \mathcal{T}^\ast ( \mathcal{T}(\sigma)^{-1} \mathcal{T}(\rho) ) \mathcal{T}^\ast ( \mathcal{T}(\rho) \mathcal{T}(\sigma)^{-1}  )\sigma\\      &\leq    \sigma \mathcal{T}^\ast ( \mathcal{T}(\sigma)^{-1} \mathcal{T}(\rho)^2 \mathcal{T}(\sigma)^{-1}  )\sigma \, ,
    \end{align}
    \end{subequations}
where in the last inequality we used the Kadison-Schwarz inequality (see, e.g., \cite[Exercise 3.4]{paulsen2002completely}). From \cite[Theorem 3.34]{HiaiMosonyi-f-divergences-2017}, the condition $\widehat{D}(\rho \| \sigma) = \widehat{D}(\mathcal{T}(\rho) \| \mathcal{T}(\sigma) ) $ is equivalent to  $\tr[\rho^2 \sigma^{-1}] = \tr[\mathcal{T}(\rho)^2 \mathcal{T}(\sigma)^{-1}]$. Thus,
\begin{subequations}
    \begin{align}
        0&=\tr[\mathcal{T}(\rho)^2 \mathcal{T}(\sigma)^{-1}]-\tr[\rho^2 \sigma^{-1}]\\
        &=\tr[\sigma  \mathcal{T}^\ast ( \mathcal{T}(\sigma)^{-1}  \mathcal{T}(\rho)^2 \mathcal{T}(\sigma)^{-1}  ) ]-\tr[\rho^2 \sigma^{-1}]\\
        &=\tr\left[\underbrace{\left(\sigma  \mathcal{T}^\ast ( \mathcal{T}(\sigma)^{-1}  \mathcal{T}(\rho)^2 \mathcal{T}(\sigma)^{-1}  )\sigma-\rho^2\right)}_{X} \sigma^{-1} \right] 
    \end{align}
\end{subequations}
which implies that $\rho^2=\sigma  \mathcal{T}^\ast ( \mathcal{T}(\sigma)^{-1}  \mathcal{T}(\rho)^2 \mathcal{T}(\sigma)^{-1}  )\sigma$, since $\sigma^{-1}$ is invertible and $X \geq 0$.
 
\vspace{0.1cm}

 \noindent  \underline{$\text{(iv)} \Rightarrow \text{(i)}. $}  Because of the condition in (iv), we have 
    \begin{equation}
        \tr[\rho^2 \sigma^{-1}] = \tr[\sigma \mathcal{T}^\ast\left(\mathcal{T}(\sigma)^{-1} \mathcal{T}(\rho)^2 \mathcal{T}(\sigma)^{-1}\right)] = \tr[\mathcal{T}(\rho)^2 \mathcal{T}(\sigma)^{-1}]
    \end{equation}
    and the proof is concluded by applying again \cite[Theorem 3.34]{HiaiMosonyi-f-divergences-2017}.
    \end{proof}

\begin{remark}
    The recoverability conditions given by Theorem \ref{theo:StructureBSDPI_appendix} are also valid for the geometric Rényi divergences, $\widehat{D}_{\alpha}(\rho\Vert \sigma)=\frac{1}{\alpha-1}\log \tr[\sigma f_{\alpha}(\sigma^{-1/2}\rho\sigma^{-1/2})]$ for $f_{\alpha}=x^{\alpha}$ \cite{fang2021geometric} with $\alpha \in (1,2]$, since this is the case where $f_{\alpha}$ is operator convex \cite[Exercise V.2.11]{bhatia-2013}. Since they are the logarithms of maximal f-divergences, and the logarithm is strictly monotone in its domain, the data-processing inequality holds for the $\widehat{D}_{\alpha}$ if and only if it holds for $\tr[\sigma f_{\alpha}(\sigma^{-1/2}\rho\sigma^{-1/2})]$.
\end{remark}

\subsection{Correspondence between states saturating BS- and relative entropy}\label{subsec:CorrespondenceEtaOmega}

Hereafter, given two states $\rho, \sigma \in \cS(\cH)$ and a quantum channel $\mathcal{T}: \cB(\cH) \rightarrow \cB(\cK) $, we say that the triple $(\rho, \sigma, \mathcal{T})$  is a \textit{Petz-triple} if it saturates the DPI for the relative entropy, namely
\begin{equation}
D(\rho \| \sigma) = D(\mathcal{T}(\rho) \| \mathcal{T}(\sigma) ) \, ,
\end{equation}
and we say that the triple $(\rho, \sigma, \mathcal{T})$ is a  \textit{BS-triple} if it saturates the DPI for the BS-entropy, namely
\begin{equation}
\widehat D(\rho \| \sigma) = \widehat D(\mathcal{T}(\rho) \| \mathcal{T}(\sigma) ) \, .
\end{equation}

Throughout the rest of the section, we denote by $\cH$ and $\mathcal{K}$ two finite-dimensional Hilbert spaces, and by $d_{\mathcal{K}}$ the dimension of $\mathcal{K}$.  Before showing the correspondence between Petz-triples and BS-triples, we will show that it holds for conditional expectations, a subclass of quantum channels. We introduce this notion by means of the proposition below \cite[Proposition 1.12]{OhyaPetz-Entropy-1993}.

\begin{proposition}

    Let $\mathcal{M}$ be a matrix algebra with unital matrix subalgebra  $\mathcal{L}$. Then, there exists a unique linear mapping $\mathcal{E}:\mathcal{M}\to \mathcal{L}$ such that 
    \begin{itemize}
        \item[1.] $\mathcal{E}$ is a positive map.
        \item[2.] $\mathcal{E}(B)=B$ for every $B \in \mathcal{L}$.
        \item[3.] $\mathcal{E}(AB)=\mathcal{E}(A)B$ for every $A \in \mathcal{M}$ and for every $B \in \mathcal{L}$.
        \item[4.]  $\mathcal{E}$ is trace preserving.
    \end{itemize}
    A map fulfilling $1$-$3$ is  called a conditional expectation.
\end{proposition}

\begin{theorem}\label{theo:theBeStTheorem}
    Let $\rho,\sigma \in \mathcal{S}(\mathcal{H})$ be quantum states satisfying supp$(\rho)\leq$supp$(\sigma)$ and let $\mathcal{E}:\mathcal{B}(\mathcal{H})\to \mathcal{B}(\mathcal{H})$ be a conditional expectation  with Stinespring's representation $\mathcal{E}(Y)=\tr_E[VYV^*]$ where $V:\mathcal{H} \to \mathcal{H}\otimes \mathcal{H}_E$ is the associated isometry. Define the states 
    \begin{equation}\label{def:GeneralEta}
        \eta_{X}:=c_{\eta_X}  \mathcal{E}(\sigma)^{-1/2}X\mathcal{E}(\sigma)^{-1/2}\, ,
    \end{equation}
    for $X\in \{\rho, \sigma,\mathcal{E}(\rho),\mathcal{E}(\sigma)\}$ and where $c_{\eta_X}$ is the normalization constant. Consider also $\rho_0=V \rho V^*$, $\sigma_0=V \sigma V^*$,  $\rho_0,\sigma_0\in \mathcal{B}(\mathcal{H}\otimes\mathcal{H}_E)^+$. If  $(\rho, \sigma, \mathcal{E})$ is a  BS-triple, then $(\eta_{\rho},\eta_{\sigma},\mathcal{E})$ is a Petz-triple, and  we can write
    \begin{equation}\label{eq:etarhoConditional}
    \eta_{\rho}=c_{\eta_{\rho}}V^*(\mathcal{E}(\sigma)^{-1/2}\otimes I_E)\rho_0(\mathcal{E}(\sigma)^{-1/2}\otimes I_E)V, \quad \text{and}  \quad  \eta_{\sigma}=c_{\eta_{\sigma}}V^*(\mathcal{E}(\sigma)^{-1/2}\otimes I_E)\sigma_0(\mathcal{E}(\sigma)^{-1/2}\otimes I_E)V\, .
\end{equation}

Conversely, let $\mu, \nu \in \mathcal{S}(\mathcal{H})$ with supp$(\mu) \leq$supp$(\nu)$ and assume that $(\mu, \nu, \mathcal{E})$ is a Petz-triple. If $\mathcal{E}(\nu)=I_{\mathcal{H}}/d_{\mathcal{H}}$, then for any positive definite $X \in \mathcal{E}(\mathcal{B}(\mathcal{H})^+)$, if we define
\begin{equation}\label{eq:DefinitionRhoSigmaConverse}
\rho:=c_{\rho}X^{1/2}\nu^{1/2} \mathcal{E}(\mu)\nu^{1/2}  X^{1/2}, \quad \text{and} \quad \sigma:=c_{\sigma}X^{1/2}\nu X^{1/2} \, ,
\end{equation}
 where $c_{\rho}$, $c_{\sigma}$ are the normalization constants,  then  $(\rho,\sigma, \mathcal{E})$ is a  BS-triple. In addition, under the constraint supp$(\mathcal{E}(\mu)) \leq$supp$(\nu)$ 
we can write 
  \begin{equation}
    \mu=d_{\mathcal{K}}c_{\rho}^{-1}V^*(X^{-1/2}\otimes I_E)\rho_0(X^{-1/2}\otimes I_E)V \, , \quad  \nu= c_{\sigma}^{-1}V^*(X^{-1/2}\otimes I_E)\sigma_0(X^{-1/2}\otimes I_E) V\, .
  \end{equation}
\end{theorem}

\begin{proof}
Let $\rho, \sigma \in \mathcal{S}(\mathcal{H})$ and let $\mathcal{E}:\mathcal{B}(\mathcal{H})\to \mathcal{B}(\mathcal{H})$ be a conditional expectation, then its range is a subalgebra and the restriction of $\mathcal{E}^*$ to $\mathcal{E}(\mathcal{B}(\mathcal{H}))$ is the natural embedding, so that we can omit it. Assume that  $(\rho,\sigma,\mathcal{E})$ is a BS-triple, then the BS-recovery condition translates into 
\begin{equation}
\begin{split}
\rho &=\sigma  \mathcal{E}(\sigma)^{-1}\mathcal{E}(\rho)\\
&=\mathcal{E}(\sigma)^{1/2}\mathcal{E}(\sigma)^{-1/2}\sigma \mathcal{E}(\sigma)^{-1/2}\mathcal{E}(\sigma)^{-1/2} \mathcal{E}(\rho)\mathcal{E}(\sigma)^{-1/2}\mathcal{E}(\sigma)^{1/2}
\end{split}
\end{equation}
Define $\Tilde{\eta}_{X}:=\mathcal{E}(\sigma)^{-1/2}X\mathcal{E}(\sigma)^{-1/2}$, so we can rewrite the expression above as
$\Tilde{\eta}_{\rho}=\Tilde{\eta}_{\sigma}\Tilde{\eta}_{\mathcal{E}(\rho)}$. Note that $[\Tilde{\eta}_{\sigma},\Tilde{\eta}_{\mathcal{E}(\rho)}]=0$, since
\begin{equation}
\Tilde{\eta}_{\sigma}\eta_{\mathcal{E}(\rho)}=\Tilde{\eta}_{\rho}=\Tilde{\eta}_{\rho}^*=\Tilde{\eta}_{\mathcal{E}(\rho)}\Tilde{\eta}_{\sigma} \, .
\end{equation}
Thus, $\Tilde{\eta}_{\rho}=\Tilde{\eta}_{\sigma}^{1/2}\Tilde{\eta}_{\mathcal{E}(\rho)}\Tilde{\eta}_{\sigma}^{1/2}$ and since $\Tilde{\eta}_{\mathcal{E}(\sigma)}=I_{\mathcal{H}}$, we obtain
\begin{equation}    \Tilde{\eta}_{\rho}=\Tilde{\eta}_{\sigma}^{1/2}\Tilde{\eta}_{\mathcal{E}(\sigma)}^{-1/2}\Tilde{\eta}_{\mathcal{E}(\rho)}\Tilde{\eta}_{\mathcal{E}(\sigma)}^{-1/2}\Tilde{\eta}_{\sigma}^{1/2} \, .
\end{equation}
We normalize now to define the states $\eta_{X}=\Tilde{\eta}_{X}/\tr[\Tilde{\eta}_{X}]$ and  notice that  $\mathcal{E}(\eta_{\rho})=\eta_{\mathcal{E}(\rho)}$ and $\mathcal{E}(\eta_{\sigma})=\eta_{\mathcal{E}(\sigma)}$ since $\mathcal{E}$ is a conditional expectation. Consequently, we can write 
\begin{equation}
\eta_{\rho}=\eta_{\sigma}^{1/2}\mathcal{E}(\eta_{\sigma})^{-1/2}\mathcal{E}(\eta_{\rho})\mathcal{E}(\eta_{\sigma})^{-1/2}\eta_{\sigma}^{1/2}\, ,
\end{equation}
so  $(\eta_{\rho},\eta_{\sigma}, \mathcal{E})$ is a Petz-triple.  

Consider now the Stinespring's isometry $V:\mathcal{H}\to \mathcal{H}\otimes \mathcal{H}_E$ such that $\mathcal{E}(X)=\tr_E[VXV^*]$, and notice that  $\rho=V^*\rho_0 V$ and $\sigma=V^*\sigma_0 V$ where $\rho_0,\sigma_0\in \mathcal{B}(\mathcal{H}\otimes\mathcal{H}_E)^+$ are positive and such that $\supp(\rho_0), \supp(\sigma_0)\le
VV^*$ . To obtain \eqref{eq:etarhoConditional}, write
\begin{equation}
    \eta_{\rho}=c_{\eta_{\rho}}\mathcal{E}(\sigma)^{-1/2}V^*\rho_0V\mathcal{E}(\sigma)^{-1/2}, \quad \text{and}  \quad  \eta_{\sigma}=c_{\eta_{\sigma}}\mathcal{E}(\sigma)^{-1/2}V^*\sigma_0V\mathcal{E}(\sigma)^{-1/2}\, .
\end{equation}
Finally, since $\mathcal{E}(\sigma)^{-1/2} \in \mathcal{E}(\mathcal{B}(\mathcal{H}))$ and the range of $\mathcal{E}$ is its multiplicative domain, Lemma \ref{lemma:multiplicative} implies that
$V\mathcal{E}(\sigma)^{-1/2}=(\mathcal{E}(\sigma)^{-1/2}\otimes I_E)V$.

Conversely, let $\mu, \nu \in \mathcal{S}(\mathcal{H})$ with supp$(\mu) \leq$supp$(\nu)$ such that  $(\mu, \nu, \mathcal{E})$ is a Petz-triple and define $\rho, \sigma$ by \eqref{eq:DefinitionRhoSigmaConverse} where $X \in \mathcal{E}(\mathcal{B}(\mathcal{H}))$. Since $(\mu,\nu, \mathcal{E})$ is a Petz-triple, it also satisfies the BS-recovery condition  \cite{BluhmCapel-BSentropy-2019}, and  using also that  $\mathcal{E}(\nu)=I_{\mathcal{H}}/d_{\mathcal{H}}$ we obtain
\begin{equation}
\nu \mathcal{E}(\mu)=\nu \mathcal{E}(\nu)^{-1}\mathcal{E}(\mu)=\mu=\mu^*=\mathcal{E}(\mu)\nu \, ,
\end{equation}
i.e. $[\nu,\mathcal{E}(\mu)]=0$.
We will show first that $(\rho,\sigma , \mathcal{E})$ is a BS-triple using the condition (iii) in Theorem \ref{thm:equivalence_recovery_conditions}. Notice that 
\begin{equation}
\mathcal{E}(\sigma)=X^{1/2}\mathcal{E}(\nu)X^{1/2}=d_{\mathcal{H}}^{-1}X, \quad  \text{and} \quad \mathcal{E}(\rho)=X^{1/2}\mathcal{E}(\nu) \mathcal{E}(\mu)X^{1/2}=d_{\mathcal{H}}^{-1}X^{1/2} \mathcal{E}(\mu)X^{1/2}.
\end{equation}
 Then,
\begin{equation}
\begin{split}
\sigma\mathcal{E}(\sigma)^{-1}\mathcal{E}(\rho)&=c_{\rho} X^{1/2}\nu X^{1/2} X^{-1}X^{1/2}\mathcal{E}(\mu) X^{1/2}\\
&=c_{\rho}X^{1/2}\nu \mathcal{E}(\mu) X^{1/2}\\
&=\rho \, ,
\end{split}
\end{equation}
so $(\rho, \sigma, \mathcal{E})$ is a BS-triple. Writing    $\sigma=V^*\sigma_0V$, then $\nu$ is given by
\begin{equation}
\begin{split}
    \nu&=c_{\sigma}^{-1}X^{-1/2}V^*\sigma_0VX^{-1/2} \\
    &=c_{\sigma}^{-1}V^*(X^{-1/2}\otimes I_E)\sigma_0(X^{-1/2}\otimes I_E) V\, ,
\end{split}
\end{equation}
where the last step follows analogously as in the previous implication  using the fact that $X \in \mathcal{E}(\mathcal{B}(\mathcal{H}))$. Analogously, write $\rho=V^* \rho_0 V$ and since supp$(\mathcal{E}(\mu))\leq$ supp$(\nu)$, we can express $\mathcal{E}(\mu)=c_{\rho}^{-1}\nu^+ X^{-1/2}\rho X^{-1/2}$, where $\nu^+$ denotes the pseudoinverse of $\nu$. Finally, since every Petz-triple $(\mu,\nu,\mathcal{E})$ is a BS-triple,  we use again (iii) of Theorem \ref{thm:equivalence_recovery_conditions} and obtain
\begin{equation}
\begin{split}
 \mu & = \nu \mathcal{E}(\nu)^{-1}\mathcal{E}(\mu)\\
 &=d_{\mathcal{H}}c_{\rho}^{-1}\nu \nu^+ X^{-1/2}\rho X^{-1/2}\\
 &=d_{\mathcal{H}}c_{\rho}^{-1}X^{-1/2}V^* \rho_0V X^{-1/2} \\
 &=d_{\mathcal{H}}c_{\rho}^{-1}V^*(X^{-1/2}\otimes I_E)\rho_0(X^{-1/2}\otimes I_E)V \, .
 \end{split}
\end{equation}
\end{proof}

\begin{remark}
    If $\mathcal{H}=\mathcal{K}\otimes \mathcal{K}^c$ notice that the partial trace $\tr_{\mathcal{K}}:\mathcal{B}(\mathcal{H})\to \mathcal{B}(\mathcal{K}^{c})$ is also included in the assumptions of Theorem \ref{theo:theBeStTheorem} setting $\mathcal{E}(X)=d_{\mathcal{K}}^{-1}I_{\mathcal{K}}\otimes \tr_{\mathcal{K}}[X]$.
\end{remark}

\begin{remark}
 The fundamental key of the previous theorem lies in the non-symmetric nature of the BS-recovering map $\mathcal{B}_{\mathcal{E}}^{\nu}$ defined in \eqref{eq:BS-recovery-condition}. This fact together with $\mathcal{E}(\nu)$ being maximally mixed impose the constraint that $[\mathcal{E}(\mu),\nu]=0$. This is an interesting  phenomenon of the BS-recovery condition that cannot be observed via the Petz recovery map.
\end{remark}

\begin{theorem}\label{thm:bs_petz}Let $\rho,\sigma$ be states on $\cB(\cH)$ and let $\mathcal{T}:B(\cH)\to B(\cK)$ be a channel.
Let $V:\cH\to \cK\otimes \cH_E$ be the isometry such that
$\mathcal{T}=\tr_E[V\cdot V^*]$.  Let us introduce the states
\[
\bar{\eta}_{\rho}:=c_\rho\mathcal{T}(\sigma)^{-1/2}V\rho
V^*\mathcal{T}(\sigma)^{-1/2},\qquad \bar{\eta}_{\sigma}:=c_\sigma\mathcal{T}(\sigma)^{-1/2}V\sigma
V^*\mathcal{T}(\sigma)^{-1/2},
\]
here $c_\rho$ and $c_\sigma$ are normalization constants. Then $(\rho,\sigma,\mathcal{T})$
 is a BS-triple if and only if $(\bar{\eta}_{\rho},\bar{\eta}_{\sigma},\tr_E)$
is a Petz-triple. 
\end{theorem}

\begin{proof} Put $\rho_0:=V\rho V^*$, $\sigma_0=V\sigma V^*$, then we clearly have that
$(\rho,\sigma,\mathcal{T})$ is a BS-triple if and only if $(\rho_0,\sigma_0,\tr_E)$ is a BS-triple. Now we can apply the results of Theorem \ref{theo:theBeStTheorem} to the latter, obtaining
the Petz-triple $(\bar{\eta}_{\rho},\bar{\eta}_{\sigma},\tr_E)$. 

Assume the converse. Note that $\tr_E[\bar{\eta}_{
\rho}]=c_\rho\mathcal{T}(\sigma)^{-1/2}\mathcal{T}(\rho)\mathcal{T}(\sigma)^{-1/2}$ and 
$\tr_E[\bar{\eta}_{\sigma}]=I_{\cK}/d_{\cK}$ is the maximally mixed state, note also that
$c_\sigma=1/d_{\cK}$. By the reversibility conditions for the triple
$(\bar{\eta}_{\rho},\bar{\eta}_{\sigma},\tr_E)$ \cite[Eq. (1.30)]{CarlenVershynina-Stability-DPI-RE-2017} or \cite{Petz2003}, we have
\[
\bar{\eta}_{\rho}^{1/2}\bar{\eta}_{\sigma}^{-1/2}=\tr_E[\bar{\eta}_{\rho}]^{1/2}\tr_E[\bar{\eta}_{\sigma}]^{-1/2}\otimes
I_E=d_{\cK}^{1/2}\tr_E[\bar{\eta}_{\rho}]^{1/2}\otimes I_E.
\]
It follows that $\bar{\eta}_{\rho}$
and $\bar{\eta}_{\sigma}$ must commute, moreover, suppressing again tensoring with the identity,
\[
\bar{\eta}_{\rho}^{1/2}=d_{\cK}^{1/2}\tr_E[\bar{\eta}_{\rho}]^{1/2}\bar{\eta}_{\sigma}^{1/2}.
\]
But this implies that $\bar{\eta}_{\sigma}$ and $\tr_E[\bar{\eta}_{\rho}]$ must commute as well
and we have
\[
\rho_0=c_\rho^{-1}\mathcal{T}(\sigma)^{1/2}\bar{\eta}_{
\rho}\mathcal{T}(\sigma)^{1/2}=c_\rho^{-1}d_{\cK}\mathcal{T}(\sigma)^{1/2}\bar{\eta}_{\sigma}\tr_E[\bar{\eta}_{\rho}]\mathcal{T}(\sigma)^{1/2}=
\sigma_0\mathcal{T}(\sigma)^{-1}\mathcal{T}(\rho)=\sigma_0\tr_E[\sigma_0]^{-1}\tr_E[\rho_0],
\]
here we used that $\mathcal{T}(\rho)=\tr_E[\rho_0]$ and similarly for $\sigma$. It follows
that $(\rho_0,\sigma_0,\tr_E)$ is a BS-triple and so is $(\rho,\sigma,\mathcal{T})$.
\end{proof}

\begin{corollary} Assume that there is some $Y\in \cB(\cH)^+$ such that $\mathcal{T}(\sigma)V=VY$
and put
\[
\eta_\rho:=d_\rho Y^{-1/2}\rho Y^{-1/2},\qquad \eta_\sigma:=d_\sigma Y^{-1/2}\sigma
Y^{-1/2},
\]
with normalization constants $d_\rho$ and $d_{\sigma}$. 
Then $(\rho,\sigma,\mathcal{T})$ is a BS-triple if and only if
$(\eta_\rho,\eta_\sigma,\mathcal{T})$ is a Petz-triple.
\end{corollary}

\begin{proof} Similarly as before, since $V\cdot V^*$ is an isometric channel,
$(\eta_\rho,\eta_\sigma,\mathcal{T})$ is a Petz-triple if and only if $(V \eta_\rho V^*,
V\eta_\sigma V^*,\tr_E)$ is a Petz-triple. Now it is enough to note that by the
assumptions, we have 
$VY^{-1/2}=\mathcal T(\sigma)^{-1/2}V$, which implies that $V\eta_\rho V^*=\bar{\eta}_\rho$
and $V\eta_\sigma V^*= \bar{\eta}_\sigma$, with the notation as in Theorem \ref{thm:bs_petz}. The statement now follows by Theorem \ref{thm:bs_petz}.
\end{proof}

\begin{remark}
The assumption in the above corollary is satisfied if $\mathcal T$ is also unital and
$\mathcal{T}(\sigma)$ lies in the image of its multiplicative domain (see Lemma
\ref{lemma:multiplicative} below). For example, this is the case for the trace preserving conditional
expectation $\mathcal E$. 
\end{remark}

\begin{corollary}\label{coro:Petz_BS_converse} Let $\mu,\nu$ be states on $\cB(\cK\otimes \cH_E)$ such that
$\tr_E[\nu]=I_{\cK}/d_{\cK}$ and assume that $(\mu,\nu,\tr_E)$ is a Petz-triple. Let $X\in
\cB(\cK)^+$ be a state and $V:\cH\to \cK\otimes\cH_E$ an isometry such that:
\begin{enumerate}
\item $X$ is invertible
\item $\supp(X^{1/2}\mu X^{1/2}), \supp(X^{1/2}\nu X^{1/2})\le VV^*$.
\end{enumerate}
Let $\omega_{\mu}(X,V), \omega_{\nu}(X,V)$ be states on $\cB(\cH)$ with $\omega_{\mu}(X,V) \propto V^*X^{1/2}\mu X^{1/2}V$,
$\omega_{\nu}(X,V)\propto V^*X^{1/2}\nu X^{1/2}V$ and let $\mathcal T=\tr_E[V\cdot V^*]$. Then
$(\omega_{\mu}(X,V),\omega_{\nu}(X,V),\mathcal T)$ is a BS-triple.

\end{corollary}

\begin{proof} Let $\rho_0$, $\sigma_0$ be states on $\cB(\cK\otimes \cH_E)$ such that
$\rho_0\propto X^{1/2}\mu X^{1/2}$ and $\sigma_0\propto X^{1/2}\nu X^{1/2}$. By the
assumptions, $V\omega_\mu(X,V) V^*=\rho_0$ and $V\omega_\nu(X,V) V^*=\sigma_0$. Further, we have 
\[
\mathcal{T}(\omega_{\nu}(X,V))=\tr_E[\sigma_0]=X,
\]
so that, with the notation as in Theorem \ref{thm:bs_petz},
\[
\bar{\eta}_\rho\propto X^{-1/2}V\omega_\mu(X,V) V^*X^{-1/2}=X^{-1/2}\rho_0X^{-1/2}\propto \mu
\]
so that $\bar{\eta}_\rho=\mu$, similarly $\bar{\eta}_\sigma=\nu$. By Theorem \ref{thm:bs_petz},
$(\omega_\mu(X,V),\omega_\nu(X,V),\mathcal{T})$ is a BS-triple.
\end{proof}

The definition of $\bar{\eta}_X$ given in Theorem \ref{thm:bs_petz} allow us to obtain a recoverability condition in a similar way as for the $\Phi_{B \to AB}$ in Corollary \ref{coro:MapPhi} for a triple $(\rho,\sigma,\mathcal{T})$ that satisifes the BS-recovery condition. 
Consider the polar decomposition
$\sigma_0^{1/2}\mathcal{T}(\sigma)^{-1/2}=c_\sigma^{-1/2}W \bar{\eta}_\sigma^{1/2}$, in the notations of
Theorem \ref{thm:bs_petz}, where $W$ is a unitary in $\cB(\cK\otimes \cH_E)$. Then $\widetilde
W:=V^*WV$ is a partial isometry in $\cB(\cH)$. Let us define the map
$\Phi_{\sigma,\mathcal{T}}\colon \cB(\cK)\to \cB(\cH)$ by
\[
\Phi_{\sigma,\mathcal{T}}(Y)=\sigma^{1/2}\widetilde
W\mathcal{T}^*(\mathcal{T}(\sigma)^{-1/2}Y\mathcal{T}(\sigma)^{-1/2})\widetilde
W^*\sigma^{1/2}.
\]
This map is obviously completely positive but not necessarily trace preserving. Note also
that if $\widetilde W=I_{\cH}$, then $\Phi_{\sigma,\mathcal{T}}$ is the Petz recovery map.  

\begin{corollary}\label{coro:recoveryPhiChannels} $(\rho,\sigma,\mathcal{T})$ is a BS-triple if and only if 
\[
\rho=\Phi_{\sigma,\mathcal{T}}(\mathcal{T}(\rho)).
\]
\end{corollary}

\begin{proof} The triple $(\rho,\sigma,\mathcal{T})$ is a BS-triple if and only if
$\rho=\sigma\mathcal{T^*}(\mathcal{T}(\sigma)^{-1}\mathcal{T}(\rho))$, equivalently,
\begin{equation}
\begin{split}
\rho&=V^*\sigma_0VV^*(\mathcal{T}(\sigma)^{-1}\mathcal{T}(\rho)\otimes I_E
)V\\
&=V^*\sigma_0(\mathcal{T}(\sigma)^{-1}\mathcal{T} (\rho)\otimes I_E)V\\
&= V^*\mathcal{T}(\sigma)^{1/2}c_\sigma^{-1}\bar{\eta}_\sigma (\tr_E[\bar{\eta}_\rho]\otimes
I_E)\mathcal{T}(\sigma)^{1/2}V\\
&= c_\sigma^{-1}V^*\mathcal{T}(\sigma)^{1/2}\bar{\eta}_\sigma^{1/2} (\tr_E[\bar{\eta}_\rho]\otimes
I_E)\bar{\eta}_\sigma^{1/2}\mathcal{T}(\sigma)^{1/2}V\\
&=V^*\sigma_0^{1/2}W(\tr_E[\bar{\eta}_\rho]\otimes
I_E)W^*\sigma_0^{1/2}V\\
&= V^*\sigma_0^{1/2}V\widetilde
W\mathcal{T}^*(\mathcal{T}(\sigma)^{-1/2}\mathcal{T}(\rho)\mathcal{T}(\sigma)^{-1/2})\widetilde
W^*V^*\sigma_0^{1/2}V.
\end{split}
\end{equation}
The proof is finished by the observation that since $\supp(\sigma_0)\le VV^*$, we have
$V^*\sigma_0^nV=(V^*\sigma_0V)^n=\sigma^n$ for any $n$, consequently,
$V^*\sigma_0^{1/2}V=\sigma^{1/2}$.
\end{proof}

\subsection{Structural Decompositions of BS- and Petz-triples satisfying recovery conditions}\label{subsec:StructuralDecomposition}

To showcase the importance of the correspondence  between BS-triples and Petz-triples given by Theorem \ref{thm:bs_petz} and Corollary \ref{coro:Petz_BS_converse}, we present below an application that consists of obtaining the structural decompositions for both quantities  in two different ways.

On the one hand, the structural decomposition for Petz-triples $(\mu, \nu, \mathcal{T})$ was fully characterized in  \cite{jenvcova2006sufficiency} and \cite{mosonyi2004structure}. In the case when $\mathcal{T}=\mathcal{E}$ is the trace-preserving conditional expectation onto some subalgebra $\mathcal{L}\subseteq \mathcal{B}(\mathcal{H})$, we  have from \cite[Theorem~5(iii)]{jenvcova2006sufficiency} that $(\rho,\sigma,\mathcal{E})$ is a Petz-triple if and only if there are density operators $\rho_1,\sigma_1\in \mathcal{L}$ and $\xi\in \mathcal{S}(\mathcal{H})$ such that 
\[
\rho=\rho_1\xi,\qquad \sigma=\sigma_1\xi, \quad \text{which implies} \quad \mathcal{E}(\rho)=\rho_1\mathcal{E}(\xi), \qquad \mathcal{E}(\sigma)=\sigma_1\mathcal{E}(\xi).
\]
Since both $\rho_1$ and $\sigma_1$ must commute with both $\xi$ and $\mathcal{E}(\xi)$, this suggests a unitary and a decomposition 
$U:\mathcal{H}\to \oplus_n \mathcal{H}_n^L\otimes \mathcal{H}_n^R$ such that
\[
\rho=U^*\left(\bigoplus_n \rho_n\otimes \xi_n\right)U,\qquad \sigma=U^*\left(\bigoplus_n \sigma_n\otimes \xi_n\right)U
\]
and 
\[
\mathcal{E}(\rho)=U^*\left(\bigoplus_n \rho_n\otimes \xi^0_n\right)U,\qquad \mathcal{E}(\sigma)=U^*\left(\bigoplus_n \sigma_n\otimes \xi^0_n\right)U
\]
for some $\rho_n,\sigma_n\in \mathcal{B}(\mathcal{H}_n^L)^+$ and $\xi_n,\xi_n^0\in \mathcal{B}(\mathcal{H}_n^R)^+$. If $\mathcal{E}(\sigma)=I_{\mathcal{H}}/d_{\mathcal{H}}$, we see that all $\sigma_n$ and $\xi_n$ must be multiples of the identity, so that we may write $\sigma=U^*\bigoplus_n (I_{\mathcal{H}^L_n}\otimes \xi_n)U$ in this case. Plugging now this structural decompositions in Corollary \ref{coro:Petz_BS_converse}, we can obtain the structural decomposition of BS-triples.

On the other hand, in the next theorem, we will show the structural decomposition for BS-recovery states directly, which will also  expand  Theorem \ref{thm:equivalence_recovery_conditions}. Using this result we will be able to obtain the structural decomposition of states of Petz triples $(\mu,\nu,\mathcal{E})$, where $\mathcal{E}$ is a conditional expectation, under the constraint that $\mathcal{E}(\nu)$ is the maximally mixed state. For this purpose, we will need the description of the multiplicative domain of a completely positive unital
map, given below in Lemma \ref{lemma:multiplicative}. 

Let $\mathcal{N}: \mathcal{B}(\mathcal{H})\to \mathcal{B}(\mathcal{K})$ be a completely positive unital map. 
From the Stinespring
representation of the adjoint map $\mathcal{N}^*$, we see that there is some auxiliary space
$\mathcal{H}_E$ and an operator $V:\mathcal{H}\to \mathcal{K}\otimes\mathcal{H}_E$ such
that $\tr_E [VV^*]=I_{\mathcal{K}}$ and $\mathcal{N}= \tr_E
[V\cdot V^*]$.  The map $\mathcal{N}$ is faithful if and only if $V^*V$ is invertible. Indeed, 
this follows from the fact that
 for any $M\ge 0$, we have $\mathcal{N}(M)=0$ if and only if $0=\tr [\mathcal{N}(M)]=\tr [MV^*V]$.
In this case, we have the polar decomposition $V=W(V^*V)^{1/2}=(VV^*)^{1/2}W$, with an
isometry $W:\mathcal{H}\to \mathcal{K}\otimes\mathcal{H}_E$.

\begin{lemma}\label{lemma:multiplicative} Let $\mathcal{N}=\tr_E[V\cdot V^*]$ be a completely positive
unital map $\mathcal{B}(\mathcal{H})\to \mathcal{B}(\mathcal{K})$ and let  $X=X^*\in \mathcal{B}(\mathcal{H})$. Then
$\mathcal{N}(X^2)=\mathcal{N}(X)^2$ if and only if
there is some $Y=Y^*\in \mathcal{B}(\mathcal{K})$ such that $(Y\otimes I_E)V=VX$. Moreover, in that case,
$Y=\mathcal{N}(X)$ and $Y\otimes I_E$ commutes with $VV^*$. If $\mathcal{N}$ is faithful, we also 
have $X=W^*(Y\otimes I_E)W$, with $W:\mathcal{H}\to \mathcal{K}\otimes\mathcal{H}_E$  the isometry from the polar decomposition
of $V$.

\end{lemma}

\begin{proof} Assume that $\mathcal{N}(X^2)=\mathcal{N}(X)^2$ and let $Y=\mathcal{N}(X)$. Let $Z=(Y\otimes
I_E)V-VX$. Then
\begin{align}
\tr[ZZ^*]&=\tr[(Y\otimes I_E)VV^*(Y\otimes I_E)-(Y\otimes I_E)VXV^*-VXV^*(Y\otimes
I_E)+VX^2V^*]\\
&=\tr[Y^2-\mathcal{N}(X^2)]=0,
\end{align}
so that $Z=0$. Conversely, let $Y=Y^*\in \mathcal{B}(\mathcal{K})$ be such that $(Y\otimes I_E)V=VX$.
Then $(Y\otimes I_E)VV^*=VXV^*=VV^*(Y\otimes I_E)$ and $Y=\tr_E[(Y\otimes
I_E)VV^*]=\tr_E[VXV^*]=\mathcal{N}(X)$. We also have
\[
\mathcal{N}(X)^2=\tr_E[(\mathcal{N}(X)\otimes I_E)^2VV^*]=\tr_E[VX^2V^*]=\mathcal{N}(X^2).
\]
Assume that $\mathcal{N}$ is faithful, then 
\[
VW^*(Y\otimes I_E)W=(VV^*)^{1/2}(Y\otimes I_E)W=(Y\otimes I_E)(VV^*)^{1/2}W=(Y\otimes
I_E)V=VX.
\]
Since $V^*V$ is invertible, it follows that we must have $X=W^*(Y\otimes I_E)W$.
\end{proof}

We will assume below that $\supp(\rho)\le \supp(\sigma)$ and use the notation
\[
[\rho/\sigma]:=\sigma^{-1/2}\rho\sigma^{-1/2}.
\]
Note that the support condition implies that this operator is well defined. 
Let $\mathcal{T}: \mathcal{B}(\mathcal{H})\to \mathcal{B}(\mathcal{K})$ be a channel and let 
$\mathcal{T}_{\sigma}$ denote the adjoint of the Petz recovery map, that is
\[
\mathcal{T}_\sigma:=(\mathcal{P}^\sigma_{\mathcal{T}})^*=\mathcal{T}(\sigma)^{-1/2}\mathcal{T}(\sigma^{1/2}\cdot\sigma^{1/2})
\mathcal{T}(\sigma)^{-1/2}.
\]
By the restriction to the supports, we may and will assume that both $\sigma$ and
$\mathcal{T}(\sigma)$
are invertible, in
which case $\mathcal{T}_\sigma: \mathcal{B}(\mathcal{H})\to \mathcal{B}(\mathcal{K})$ is unital and faithful. It is also easily seen that
\[
\mathcal{T}_\sigma([\rho/\sigma])=[\mathcal{T}(\rho)/\mathcal{T}(\sigma)].
\]

\begin{theorem}\label{theo:StructureBSDPI_appendix} Let $\mathcal{T}:\mathcal{B}(\mathcal{H})\to
\mathcal{B}(\mathcal{K})$  be a channel  and let $V: \mathcal{H}\to\mathcal{K}\otimes\mathcal{H}_E$ be an isometry such that 
$\mathcal{T}=\tr_E[V\cdot
V^*]$.  Let $\rho,\sigma\in \mathcal{S}(\mathcal{H})$ be states such that $\supp(\rho)\le \supp(\sigma)$. The following conditions are
equivalent.
\begin{enumerate}
\item[(i)] $\widehat D(\rho\|\sigma)=\widehat D(\mathcal{T}(\rho)\|\mathcal{T}(\sigma))$.
\item[(ii)]  $\mathcal{T}_\sigma([\rho/\sigma]^2)=\mathcal{T}_\sigma([\rho/\sigma])^2$.
\item[(iii)] There is a decomposition and a unitary $U:\mathcal{K}\to \bigoplus_n
\mathcal{K}_n^L\otimes \mathcal{K}_n^R$, such that for
\[
\rho=V^*\rho_0V,\quad \sigma=V^*\sigma_0 V,
\]
where $\rho_0,\sigma_0\in \mathcal{B}(\mathcal{K}\otimes\mathcal{H}_E)^+$ are positive and such that $\supp(\rho_0), \supp(\sigma_0)\le
VV^*$, we have 
\begin{subequations}
\begin{align}\label{eq:Structurerho0T}
\rho_0&=(\mathcal{T}(\sigma)^{1/2}U^*\otimes I_E)\bigoplus_n (\xi_n^L\otimes
\xi_n^R)(U\mathcal{T}(\sigma)^{1/2}\otimes I_E)
\end{align}
\begin{align}\label{eq:Structuresigma0T}
\sigma_0&=(\mathcal{T}(\sigma)^{1/2}U^*\otimes I_E)\bigoplus_n (I_{\mathcal{K}_n^L}\otimes
\xi_n^R)(U\mathcal{T}(\sigma)^{1/2}\otimes I_E)
\end{align}
\end{subequations}
for some $\xi_n^L\in \mathcal{B}(\mathcal{K}_n^L)^+$ and $\xi_n^R\in
\mathcal{B}(\mathcal{K}_n^R\otimes\mathcal{H}_E)^+$.
\item[(iv)] $\rho=\sigma\mathcal{T}^*(\mathcal{T}(\sigma)^{-1}\mathcal{T}(\rho))$.
\end{enumerate}

\end{theorem}

\begin{proof} 
 \noindent \underline{$\text{(i)} \Leftrightarrow \text{(ii)} .$} This equivalence was proved in \cite{HiaiMosonyi-f-divergences-2017}, in a more general situation. To make the proof self-contained, we present a proof under our assumptions. First, we note that we may write
\[
\widehat D(\rho\|\sigma)=\tr[\sigma^{1/2}\rho\sigma^{-1/2}\log(\sigma^{-1/2}\rho\sigma^{-1/2})]=\tr[\sigma f([\rho/\sigma])],
\]
with $f(x)=x\log x$. Using the integral representation
\[
f(x)=\int_0^\infty \left(\frac{x}{1+t}-\frac{x}{x+t}\right)dt=\int_0^\infty
\left(\frac{x}{1+t}-1+\frac{t}{x+t}\right)dt,
\]
we obtain
\[
\widehat
D(\rho\|\sigma)=\int_0^\infty \left(\frac{\tr[\rho]}{1+t}-\tr[\sigma]+t\tr[\sigma([\rho/\sigma]+t)^{-1}] \right)dt.
\]
The equality (i) holds if and only if
\begin{equation}\label{eq:int}
\int_0^\infty
t(\tr[\sigma([\rho/\sigma]+t)^{-1}]-\tr[\mathcal{T}(\sigma)([\mathcal{T}(\rho)/\mathcal{T}(\sigma)]+t)^{-1}])dt=0.
\end{equation}
Now note that
\[
\tr[\sigma([\rho/\sigma]+t)^{-1}]=\tr[\mathcal{T}(\sigma)\mathcal{T}_\sigma(([\rho/\sigma]+t)^{-1})]
\]
and
\[
\tr[\mathcal{T}(\sigma)([\mathcal{T}(\rho)/\mathcal{T}(\sigma)]+t)^{-1}]=\tr[\mathcal{T}(\sigma)(\mathcal{T}_\sigma([\rho/\sigma]+t))^{-1}].
\]
As mentioned above,  we may assume that $\sigma$ and $\mathcal{T}(\sigma)$ are invertible
and then 
$\mathcal{T}_\sigma$ is unital and faithful. It
follows by the Choi inequality \cite{choidavis, choi1974schwarz} that Eq.\ \eqref{eq:int} holds if and only if  
\begin{equation}\label{eq:multiplicativet}
\mathcal{T}_\sigma(([\rho/\sigma]+t)^{-1})=(\mathcal{T}_\sigma([\rho/\sigma]+t))^{-1},\qquad \forall
t\in (0,\infty).
\end{equation}
Differentiating by $t$, we obtain 
\[
\mathcal{T}_\sigma(([\rho/\sigma]+t)^{-2})=(\mathcal{T}_\sigma([\rho/\sigma]+t))^{-2}=\mathcal{T}_\sigma(([\rho/\sigma]+t)^{-1})^2,
\]
where we have used Eq.\ \eqref{eq:multiplicativet} in the last equality. Hence, $\mathcal{T}_\sigma$ must be multiplicative on all elements of
the form $([\rho/\sigma]+t)^{-1}$, $t\in (0,\infty)$. Since the multiplicative domain of $\mathcal{T}_\sigma$ is a subalgebra, we see that the equality Eq.\ \eqref{eq:multiplicativet}, and
hence also (i), is equivalent to (ii). 

\vspace{0.2cm}

\noindent \underline{$\text{(ii)} \Rightarrow \text{(iii)} .$} 
 Assume (ii) and put
$N:=\mathcal{T}_\sigma([\rho/\sigma])=[\mathcal{T}(\rho)/\mathcal{T}(\sigma)]$. We have  $\mathcal{T}_\sigma=\tr_E[S\cdot S^*]$, with
$S=(\mathcal{T}(\sigma)^{-1/2}\otimes I_E)V\sigma^{1/2}$ and using Lemma \ref{lemma:multiplicative}, we
see  that $N\otimes I_E$  must commute with
$M:=SS^*$. Since $\mathcal{T}_\sigma$ is faithful, we have 
$[\rho/\sigma]=W^*(N\otimes I_E)W$, where $S=M^{1/2}W$ is the polar decomposition. 
By definition of $S$, we obtain  $\sigma^{1/2}=V^*(\mathcal{T}(\sigma)^{1/2}\otimes I_E)S$, so that
\[
\sigma=V^*(\mathcal{T}(\sigma)^{1/2}\otimes I_E)M(\mathcal{T}(\sigma)^{1/2}\otimes I_E)V
\]
and
\begin{align}
\rho=\sigma^{1/2}[\rho/\sigma]\sigma^{1/2}&=V^*(\mathcal{T}(\sigma)^{1/2}\otimes
I_E)M^{1/2}(N\otimes I_E)M^{1/2}(\mathcal{T}(\sigma)^{1/2}\otimes
I_E)V\\
&=V^*(\mathcal{T}(\sigma)^{1/2}\otimes
I_E)(N\otimes I_E)M(\mathcal{T}(\sigma)^{1/2}\otimes
I_E)V.
\end{align}
Let $\mathcal{A}\subseteq \mathcal{B}(\mathcal{K})$ be the unital subalgebra generated by $N$, then there is a
decomposition and unitary $U$ as in the statement (iii), such that 
\[
\mathcal{A}=U^*\bigoplus_n(\mathcal{B}(\mathcal{K}_n^L)\otimes I_{\mathcal{K}_n^R})U.
\]
Since $M\in (\mathcal{A}\otimes I_E)'=\mathcal{A}'\otimes \mathcal{B}(\mathcal{H}_E)$, we obtain the
decompositions
\[
N=U^*\bigoplus_n (\xi_n^L\otimes I_{\mathcal{K}_n^R})U,\quad M=(U^*\otimes
I_E)\bigoplus_n (I_{\mathcal{K}_n^L}\otimes \xi_n^R)(U\otimes I_E)
\]
for some $\xi_n^L\in \mathcal{B}(\mathcal{K}_n^L)^+$ and $\xi_n^R\in \mathcal{B}(\mathcal{K}_n^R \otimes \mathcal{H}_E)^+$. 
Now put 
\begin{align}
\rho_0&:=(\mathcal{T}(\sigma)^{1/2}\otimes
I_E)(N\otimes I_E)M(\mathcal{T}(\sigma)^{1/2}\otimes
I_E)\\
\sigma_0&:=(\mathcal{T}(\sigma)^{1/2}\otimes I_E)M(\mathcal{T}(\sigma)^{1/2}\otimes I_E).
\end{align}
The only thing left to prove is the condition on the supports. Put $P:=VV^*$. By definition of $\sigma_0$
and $M=SS^*$, we obtain $\sigma_0 =V\sigma V^*$, so that $P\sigma_0 P=V\sigma
V^*=\sigma_0$. We also have 
\[
\rho_0=(\mathcal{T}(\sigma)^{1/2}\otimes
I_E)(N\otimes I_E)M(\mathcal{T}(\sigma)^{1/2}\otimes
I_E)=(\mathcal{T}(\sigma)^{1/2}\otimes
I_E)(N\otimes I_E)S\sigma^{1/2}V^*=\rho_0P.
\]
This proves the assertion (iii).

\vspace{0.2cm}

\noindent \underline{$\text{(iii)} \Rightarrow \text{(iv)} .$}  If (iii) holds, then 
\[
\mathcal{T}(\sigma)=\tr_E[\sigma_0]=\mathcal{T}(\sigma)^{1/2}U^*\bigoplus_n
(I_{\mathcal{K}_n^L}\otimes
\tr_E[\xi_n^R])U\mathcal{T}(\sigma)^{1/2}.
\]
Since $\mathcal{T}(\sigma)$ is invertible and $U$ unitary, this implies that
$\tr_E[\xi_n^R]=I_{\mathcal{K}_n^R}$, so that
\[
\mathcal{T}(\rho)=\tr_E[V\rho V^*]=\tr_E[\rho_0]=\mathcal{T}(\sigma)^{1/2}U^*\bigoplus_n (\xi_n^L\otimes
I_{\mathcal{K}_n^R})U\mathcal{T}(\sigma)^{1/2}.
\]
Since $\mathcal{T}^*= V^*(\cdot \otimes I_E)V$, it follows that
\begin{align}
\sigma\mathcal{T}^*(\mathcal{T}(\sigma)^{-1}\mathcal{T}(\rho))&=\sigma\mathcal{T}^*(\mathcal{T}(\sigma)^{-1/2}U^*\bigoplus_n 
(\xi_n^L\otimes I_{\mathcal{K}_n^R})U\mathcal{T}(\sigma)^{1/2})\\
&=V^*\sigma_0VV^*(\mathcal{T}(\sigma)^{-1/2}U^*\bigoplus_n (\xi_n^L\otimes
I_{\mathcal{K}_n^R})U\mathcal{T}(\sigma)^{1/2}\otimes I_E)V\\
&=V^*\sigma_0(\mathcal{T}(\sigma)^{-1/2}U^*\otimes I_E)\bigoplus_n (\xi_n^L\otimes
I_{\mathcal{K}_n^R\otimes\mathcal{H}_E})(U\mathcal{T}(\sigma)^{1/2}\otimes I_E)V\\
&=V^*(\mathcal{T}(\sigma)^{1/2}U^*\otimes I_E)\bigoplus_n (\xi_n^L\otimes
\xi_n^R)(U\mathcal{T}(\sigma)^{1/2}\otimes I_E)V=\rho,
\end{align}
which is (iv).

\vspace{0.2cm}

\noindent \underline{$\text{(iv)} \Rightarrow \text{(ii)} .$}  To finish the proof, assume (iv). Then we have
\begin{align}
\tr[\mathcal{T}(\sigma)\mathcal{T}_\sigma([\rho/\sigma])^2]&=\tr[\mathcal{T}(\rho)\mathcal{T}(\sigma)^{-1}\mathcal{T}(\rho)]=
\tr[\rho\mathcal{T}^*(\mathcal{T}(\sigma)^{-1}\mathcal{T}(\rho))]=\tr[\rho^2\sigma^{-1}]\\
&=\tr[\mathcal{T}(\sigma)\mathcal{T}_\sigma([\rho/\sigma]^2)].
\end{align}
Since $\mathcal{T}(\sigma)$ is faithful and we always have $\mathcal{T}_\sigma([\rho/\sigma])^2\le
\mathcal{T}_\sigma([\rho/\sigma]^2)$ by the Kadison-Schwarz inequality, this implies (ii). 
\end{proof}

As a consequence of Theorem \ref{theo:theBeStTheorem} and Theorem \ref{theo:StructureBSDPI_appendix} we obtain the structural decomposition for states that satruate the BS-relative entropy for conditional expectations.
\begin{corollary}
Let $\mu,\nu$ be states on $\cS(\cK\otimes \cH_E)$ and $\mathcal{E}$ a conditional expectation such that $(\mu,\nu,\mathcal{E})$ satisfies the Petz recovery condition. If $\mathcal{E}(\nu)=I_{\mathcal{K}}/d_{\mathcal{K}}$, then for any positive definite $X \in \mathcal{B}(\mathcal{K})$, if we define
 $c_{\rho}=\tr[X^{1/2} \nu^{1/2} \mathcal{E}(\mu)\nu^{1/2} X^{1/2}]^{-1}$ and $c_{\sigma}=\tr[X^{1/2}\nu X^{1/2}]^{-1}$,   there is a decomposition and a unitary $U:\mathcal{K}\to \bigoplus_n
\mathcal{K}_n^L\otimes \mathcal{K}_n^R$, such that
 \begin{subequations}
 \begin{align}
\mu &=d_{\mathcal{K}}c_{\rho}^{-1}V^* (X^{-1/2}\mathcal{E}(\sigma)^{1/2}U^*\otimes I_E)\bigoplus_n (\xi_n^L\otimes
\xi_n^R)(U\mathcal{E}(\sigma)^{1/2}X^{-1/2}\otimes I_E) V \, ,
\end{align}
 \begin{align}
\nu &=c_{\sigma}^{-1}V^* (X^{-1/2}\mathcal{E}(\sigma)^{1/2}U^*\otimes I_E)\bigoplus_n (I_{\mathcal{K}_n^L}\otimes
\xi_n^R)(U\mathcal{E}(\sigma)^{1/2}X^{-1/2}\otimes I_E) V \, .
\end{align}
 \end{subequations}
for some $\xi_n^L\in \mathcal{B}(\mathcal{K}_n^L)^+$ and $\xi_n^R\in
\mathcal{B}(\mathcal{K}_n^R\otimes\mathcal{H}_E)^+$ and  where $\sigma=V \sigma_0 V^*$, $\rho=V \rho_0 V^*$ and $\rho_0$, $\sigma_0$ have a decomposition of the form of \eqref{eq:Structurerho0T} and $\eqref{eq:Structuresigma0T}$, respectively. In particular, if we can take $X=\mathcal{E}(\sigma)^{-1/2}$  we obtain
 \begin{subequations}
 \begin{align}
\mu &=d_{\mathcal{K}}c_{\rho}^{-1}V^* (U^*\otimes I_E)\bigoplus_n (\xi_n^L\otimes
\xi_n^R)(U\otimes I_E) V \, ,
\end{align}
 \begin{align}
\nu &=c_{\sigma}^{-1}V^* (U^*\otimes I_E)\bigoplus_n (I_{\mathcal{K}_n^L}\otimes
\xi_n^R)(U\otimes I_E) V \, .
\end{align}
 \end{subequations}
\end{corollary}

\section{Applications in the context of quantum spin systems}

\subsection{Superexponential conditional independence of quantum spin chains}

In this section we provide some applications of the results derived in the previous pages in the context of quantum spin systems.  Quantum spin systems are mathematical models that describe arrays of atoms and their interactions.The mathematical formalism that typically describes these systems and the operators defined over them is as follows.  For any finite subset $\Lambda \subset \mathbb{Z}$, $\vert \Lambda \vert < \infty$, we associate a finite-dimensional Hilbert space $\mathcal{H}_{\Lambda}=\otimes_{x \in \Lambda}\mathcal{H}_x$, where $\mathcal{H}_x=\mathbb{C}^d$. The algebra of bounded linear operators on $\mathcal{H}_{\Lambda}$ is then given by $\mathcal{A}_{\Lambda}=\mathcal{B}(\mathcal{H}_{\Lambda})$, which has a  $C^*$-algebra structure. When considering a subset $\Lambda' \subseteq \Lambda$, there is a natural embedding  $\mathcal{A}_{\Lambda'}\subseteq \mathcal{A}_{\Lambda}$: given $X \in \mathcal{A}_{\Lambda'}$, we identify it with $X \otimes I_{\Lambda\setminus \Lambda'} \in\mathcal{A}_{\Lambda}$. Consequently, it is possible to define the algebra of quasi-local observables for a general  set $\Sigma \subseteq \mathbb{Z}$ as the closure of the union of the local algebras with respect to the operator norm,
\begin{equation}
    \mathcal{A}_{\Sigma}=\overline{\bigcup_{\substack{\Lambda\subseteq \Sigma\\ \vert \Lambda\vert <\infty}}\mathcal{A}_{\Lambda}}^{\Vert \hspace{2pt}\cdot \hspace{2pt}\Vert_{\infty}}.
\end{equation}
An interaction on a quantum spin system $\Phi$ is an application that maps any finite set $\Lambda$ into the algebra $\mathcal{A}_{\Lambda}$ satisfying $\Phi(\Lambda)=\Phi(\Lambda)^*$, and has as  local Hamiltonian
\begin{equation}
    H_{\Lambda}=\sum_{\Lambda' \subseteq \Lambda} \Phi(\Lambda').
\end{equation}
A local Hamiltonian on the finite set $\Lambda$ is said to have finite range if there exists two constants $R,J>0$ such that
\begin{itemize}
    \item[(i)] $\Phi(\Lambda)=0$ whenever $\operatorname{diam}(\Lambda):=\max\{x-y:x,y \in \Lambda\}> R$,
    \item[(ii)] For every finite set $\Lambda \subset \mathbb{Z}$, $\Vert \Phi(\Lambda)\Vert_\infty\leq J$. 
\end{itemize}
 We give several applications of our results concerning QMCs and BS-QMCs in this section. The first one concerns Gibbs states of local Hamiltonians and their conditional independence. Consider a finite interval of $\mathbb{Z}$, $I \subset \mathbb{Z}$ split into $I=ABC$ as in \Cref{fig:1}, and a local Hamiltonian $H_{ABC}$ on it (i.e. a self-adjoint operator satisfying $H_{ABC}=\sum_{X\subset I} H_X$).  The Gibbs state of $H_{ABC}$ at inverse temperature $\beta < \infty $ in then given by $e^{-\beta H_{ABC}}/ \Tr[e^{-\beta H_{ABC}}]$. Studying their conditional independence, i.e. how correlated regions $A$ and $C$ are conditioned on $B$, is a fundamental problem in quantum spin systems. In \cite{brown-2012}, it was proven that a state $\sigma_{ABC}$ can be written as Gibbs state of a local commuting Hamiltonian (in which $[H_X, H_Y]=0$ for every $X, Y \subset I$) if, and only if, $\sigma_{ABC}$ is a QMC between $A \leftrightarrow B \leftrightarrow C$. Therefore, for such states, the CMI (and the BS-CMI) vanish, and we say that they are fully conditional independent. When the Hamiltonian considered is not commuting, though, the situation is much more subtle. 

In \cite{Kato.2019}, it was proven that for Gibbs states of local, finite-range, translation-invariant Hamiltonians in 1D at any positive temperature, the CMI decays subexponentially with the size of B, and this was subsequently improved to exponential decay in  \cite{Kuwahara.2024}, while for the case of the BS-CMI it was proven in \cite{gondolf2024conditional} that the decay is superexponential. It is then a natural question whether the same superexponential behaviour can be proven for the CMI as well. It seems that a possible way would be by relating the BS-CMI of a certain Gibbs state to the CMI of another.  Throughout this work we have exploited the correspondence between BS-QMCs and QMCs given by $\rho_{ABC}\leftrightarrow \eta_{ABC}$ and have also provided bounds for their approximative versions in Proposition \ref{prop:relation_Irev_I}. We now employ that connection to prove  superexponential decay of the CMI of $ \eta_{ABC}$ whenever $\rho_{ABC}$ is a Gibbs state of a local, finite-range, translation-invariant Hamiltonian, exploiting the fact that these Gibbs states are approximate BS-QMCs proven in \cite{gondolf2024conditional}.

\begin{figure}[ht]
\begin{center}

\begin{tikzpicture}[scale=0.7]

\definecolor{frenchblue}{rgb}{0.0, 0.45, 0.73}
\Block[5,blue!50!white,A,1];

\begin{scope}[xshift=5cm]
\Block[5,teal!50!white,B,1];
\end{scope}
\begin{scope}[xshift=10cm]
\Block[5,green!50!white,C,1];
\end{scope}

\node at (15,0.8) {\huge $I$};

\end{tikzpicture}

  \caption{An interval $I$ split into three subintervals $I=ABC$ such that $B$ shields $A$ from $C$. }
  \label{fig:1}
      
\end{center}
\end{figure}

\begin{theorem}[Superexponential decay of CMI for $\eta_{ABC}$]\label{thm:superexponential_decay}
    Let us consider a quantum spin system on $\mathbb{Z}$ with local, finite-range, translation-invariant interactions and $\rho_{ABC}$ the associated Gibbs state at inverse temperature $\beta < \infty$. Then, there exists a positive function $l \mapsto \varepsilon(l)$ with  superexponential decay such that for every finite interval $I \subset \mathbb{Z}$ split into three subintervals $I = ABC$ where $B$ shields $A$ from $C$ (see Figure \ref{fig:1}) and $\eta_{ABC}=\frac{1}{d_B}\rho_B^{-1/2}\rho_{ABC}\rho_B^{-1/2}$, 
\begin{equation}
    I_{\eta}(A:C\vert B)\leq \mathcal{C} r(d_A,d_B,d_C)^{1/2}\Vert \rho_B^{-1}\Vert_{\infty}^{1/2}e^{\alpha \vert A \vert}\varepsilon(\vert B\vert),
\end{equation}
where $r$ is a rational function,  $\mathcal{C}$ and $\alpha$ are constants only depending on inverse temperature $\beta$, strength $J$ and range $R$ of the potential (see \cite[Section 2.4]{gondolf2024conditional}). 
\end{theorem}
\begin{proof}
     Let $\eta_{ABC}=\frac{1}{d_B}\rho_B^{-1/2}\rho_{ABC}\rho_B^{-1/2}$ be the state associated to the Gibss state $\rho_{ABC}$ at inverse temperature $\beta$ under the conditions of the statement, i.e. $\rho_{ABC}=e^{-\beta H_{ABC}}/\tr[e^{-\beta H_{ABC}}]$. From Proposition  \ref{prop:relation_Irev_I}, we can upper bound the conditional mutual information for $\eta_{ABC}$, $I_{\eta}(A:C\vert B)$, in terms of the reversed BS-conditional mutual information of $\rho_{ABC}$  together with other extra terms as follows 
    \begin{equation}
        I_{\eta}(A:C\vert B)\leq 4\sqrt{\frac{2(d_A+ d_C+1)^2}{d_B\pi }}\left\Vert  \rho_B^{-1} \right\Vert_{\infty}^{1/2} \left\Vert \rho_{BC}^{-1/2}\rho_{ABC}\rho_{BC}^{-1/2}\right\Vert_{\infty}^{1/4}\widehat{I}^{\mathrm{rev}}_{\rho}(A:C | B)^{1/8}\, .
    \end{equation}
    Let us show now that each of the multiplicative terms in the RHS above depending on marginals of $\rho_{ABC}$ grows at most exponentially with $|B|$, and the last term decays superexponentially with $|B|$, giving us the right decay. Define now, for any consecutive $X,Y \subset ABC$, the operator $E_{XY}=e^{-H_{XY}}e^{H_X+H_Y}$, known as the \textit{Araki expansional} \cite{Araki-Gibbs-states-1D-1969}, where we are omitting $\beta$ in the exponentials for simplicity. By a very similar calculation to that of \cite[Lemma 3.5]{gondolf2024conditional}, we have
    \begin{equation}
    \left\Vert  \rho_B^{-1} \right\Vert_{\infty}^{1/2} \leq \mathcal{C}_1 e^{\alpha_1 \vert B \vert} \, .
    \end{equation}
    We reproduce the calculations here for completeness. Using the notation above, and denoting $Z_{X}= \tr[ e^{-H_{X}}]$, and $\rho^X=e^{-H_X}/ \tr[e^{-H_X}]$ for any $X\subset ABC$,
    \begin{align}
    \rho_B^{-1} & = \tr_{AC}[e^{-H_{ABC}}]^{-1} Z_{ABC} \\
    & =  \tr_{AC}[e^{-H_A-H_B-H_C}e^{H_A+H_B+H_C}e^{-H_{ABC}}]^{-1} Z_{ABC} \\
    & = \tr_{AC}[\rho^A \otimes \rho^C E_{A,B}^T E_{AB,C}^T ]^{-1} (\rho^B)^{-1} \frac{Z_{ABC}}{Z_A Z_B Z_C} \, ,
    \end{align}
    where the last term can also be rewritten as 
    \begin{equation}
    \frac{Z_{ABC}}{Z_A Z_B Z_C} = \tr[\rho^A \otimes \rho^B \otimes \rho^C E_{A,B}^T E_{AB,C}^T ] \, .
    \end{equation}
    Therefore, by \cite[Corollary 4.4]{BluhmCapelPerezHernandez-ExpDecayMI-2021}, we conclude
    \begin{equation}
\norm{(\rho_B)^{-1}}_\infty \leq \mathcal{C} \norm{(\rho^B)^{-1}}_\infty \, .
    \end{equation}
    The fact that the last term in the RHS above scales exponentially with $|B|$ follows directly from $H_B$ being a local, finite-range Hamiltonian, by bounding 
    \begin{equation}
    |B|\,  \underset{X\subset B}{\operatorname{min}} \,  \underset{i\subset I}{\operatorname{min}} \,  |\lambda_i (H_X)|\leq \norm{H_B}_\infty \leq |B|\, \underset{X\subset B}{\operatorname{max}}\, \underset{i\subset I}{\operatorname{max}} \,|\lambda_i (H_X)| \, .
    \end{equation}
     Afterwards,  \cite[Theorem 3.6]{gondolf2024conditional} states that the reversed conditional mutual information of the Gibbs state decays superexponential fast with the size of $B$. Concretely,  there exists a function  $\varepsilon(\hspace{2pt}\cdot \hspace{2pt})$ with superexponential decay, positive constants $ \alpha_2,\mathcal{C}_2$  such that
    \begin{equation}
        \widehat{I}^{\mathrm{rev}}_{\rho}(A;C | B)^{1/8}\leq \mathcal{C}_2e^{\alpha_2\vert A\vert}\varepsilon(\vert B \vert),
    \end{equation}
    with $\varepsilon(\hspace{2pt}\cdot \hspace{2pt})$ a positive function with superexponential decay. In order to bound the last term, we use  first \cite[Theorem IX.1.1]{bhatia-2013} and obtain 
    \begin{equation}
        \left\Vert \rho_{BC}^{-1/2}\rho_{ABC}\rho_{BC}^{-1/2}\right\Vert_{\infty}\leq \left\Vert \rho_{BC}^{-1}\rho_{ABC}\right\Vert_{\infty} \, .
    \end{equation}
  Finally, we reproduce the computations in the proof of \cite[Theorem 5.1]{BluhmCapelPerezHernandez-ExpDecayMI-2021} to bound this last term. We first write 
  \begin{equation}
  \begin{split}
    \rho_{BC}^{-1} \rho_{ABC}&= \tr_A(e^{-H_{ABC}}) e^{-H_{ABC}}\\
    &=\tr_A[ e^{-H_A}e^{H_A+H_{BC}}e^{-H_{ABC}}]^{-1} e^{-H_A} e^{H_A+H_{BC}}e^{-H_{ABC}}\\
    &= \tr_A [ \rho^A E_{A,BC}^T]^{-1} \rho^{A} E_{A,BC}^T \, .
    \end{split}
  \end{equation}
  In \cite[Corollary 3.4 (i)]{BluhmCapelPerezHernandez-ExpDecayMI-2021}, it is shown that there exists a constant $\mathcal{C}_3$ such that $\Vert E_{A,BC} \Vert_\infty \leq \mathcal{C}_3$, and following an analogous proof to that of  \cite[Corollary 4.4]{BluhmCapelPerezHernandez-ExpDecayMI-2021}, we can bound $\norm{\tr_A( \rho^A E_{A,BC}^T)^{-1}}_\infty \leq \mathcal{C}_4$. Finally, similarly as above for $\norm{(\rho^B)^{-1}}_\infty$, we can bound $\norm{\rho^A}_\infty$ by an exponential factor in $|A|$. 
\end{proof}

The significance of this result is as follows. As mentioned in \Cref{rem:Gibbs_comm_QMC}, Gibbs states of local, commuting Hamiltonians are quantum Markov chains, and thus their CMI vanishes. All prior examples regarding decay of CMI, particularly for Gibbs states, have shown that, whenever it does not vanish, its decay with $|B|$ is at most exponential. Our result therefore gives the first examples of states with a faster decay of CMI (without it vanishing) and provides a large family of states that lies between those that are fully conditionally independent, such as Gibbs states of local, commuting Hamiltonians, and with CMI decaying exponentially with  $|B|$, such as as Gibbs states of local Hamiltonians. However, the precise physical interpretation of our $\eta_{ABC}$ in the context of Gibbs states is unfortunately still an open question that we leave for future work.

\subsection{Decay of correlations in quantum spin systems}

 In \cite{lucia2025spectral}, the authors consider a quantum spin system in a finite volume $\Lambda$, and for any $4$-partition of it $ABCD$, they define the quantity 
\begin{equation}
\Delta_\rho (A:C | D) := \underset{R_{AD},Q_{CD}}{\operatorname{sup}} \left| \tr_{ACD}[(\rho_{ACD} - \rho_{AD}\rho^{-1}_D\rho_{CD}) Q_{CD}^* R_{AD}] \right|
\end{equation}
for $\rho \in \cS(\cH_{ABCD})$, where the supremum is taken over $R_{AD}\in \cB(\cH_{AD})$, $Q_{CD}\in \cB(\cH_{CD})$ such that $\tr[ \rho \, R^*_{AD} R_{AD}] = \tr[\rho \, Q^*_{CD} Q_{CD} ]= 1$. The typical geometry that they consider is that of \Cref{fig:2}. The main finding of this paper is that a sufficiently fast decay of $\Delta_\rho$ with respect to the distance between $A$ and $C$ is equivalent, under certain technical assumptions, to a positive spectral gap of a Davies generator with unique fixed point $\rho$. Note that the Davies generator is the standard model to describe the evolution of physical spin systems weakly coupled to an environment.

\begin{figure}[ht]
\begin{center}

\begin{tikzpicture}[scale=0.5]

\fill [blue!50!white] (-0.5,-0.5) rectangle (18.5,0.5);
\foreach \m in {1,...,19}{
\shade[shading=ball, ball color=orange] (\m-1,0) circle (0.3);
}

\begin{scope}[yshift=-1cm]
\fill [blue!50!white] (-0.5,-4.5) rectangle (1.5,0.5);
\foreach \m in {1,...,2}{
\foreach \n in {1,...,5}{
\shade[shading=ball, ball color=orange] (\m-1,\n-5) circle (0.3);
}
}

\end{scope}

\begin{scope}[yshift=-1cm,xshift=17cm]
\fill [blue!50!white] (-0.5,-4.5) rectangle (1.5,0.5);
\foreach \m in {1,...,2}{
\foreach \n in {1,...,5}{
\shade[shading=ball, ball color=orange] (\m-1,\n-5) circle (0.3);
}
}

\end{scope}

\begin{scope}[yshift=-6cm]
\fill [blue!50!white] (-0.5,-0.5) rectangle (18.5,0.5);
\foreach \m in {1,...,19}{
\shade[shading=ball, ball color=orange] (\m-1,0) circle (0.3);
}
\end{scope}

\begin{scope}[yshift=-1cm,xshift=2cm]
\fill [teal!50!white] (-0.5,-4.5) rectangle (4.5,0.5);
\foreach \m in {1,...,5}{
\foreach \n in {1,...,5}{
\shade[shading=ball, ball color=orange] (\m-1,\n-5) circle (0.3);
}
}

\end{scope}

\begin{scope}[yshift=-1cm,xshift=7cm]
\fill [green!50!white] (-0.5,-4.5) rectangle (4.5,0.5);
\foreach \m in {1,...,5}{
\foreach \n in {1,...,5}{
\shade[shading=ball, ball color=orange] (\m-1,\n-5) circle (0.3);
}
}

\end{scope}

\begin{scope}[yshift=-1cm,xshift=12cm]
\fill [yellow!50!white] (-0.5,-4.5) rectangle (4.5,0.5);
\foreach \m in {1,...,5}{
\foreach \n in {1,...,5}{
\shade[shading=ball, ball color=orange] (\m-1,\n-5) circle (0.3);
}
}

\end{scope}

\node at  (3.5,-3.3) {\Huge $\textbf{A}$};

\node at  (8.5,-3.3) {\Huge $\textbf{B}$};

\node at  (13.5,-3.3) {\Huge $\textbf{C}$};

\node at  (17.5,-1.3) {\Huge $\textbf{D}$};

\end{tikzpicture}

  \caption{$\Lambda$ split into four subsystems $ABCD$ such that $B$ shields $A$ from $C$ and $D$ surrounds them.  }
  \label{fig:2}
      
\end{center}
\end{figure}

From \Cref{theo:StructureBSDPI}, it is obvious that for BS-QMCs $\rho$ between $A \leftrightarrow D \leftrightarrow C$, $\Delta_\rho (A:C | D)=0$. In \cite{lucia2025spectral} the converse is proven to be true as well. Since a vanishing $\Delta_\rho$ is an equivalent condition to BS-QMCs, it is natural to think that approximate BS-QMCs will have a small $\Delta_\rho$. We show this below.

\begin{proposition}\label{prop:Delta_BSQMC}
In the conditions above, for $\rho \in \cS(\cH_{ABCD})$, we have 
\begin{equation}
  \Delta_\rho (A:C | D) \leq f(\rho) \widehat{I}^{\operatorname{rev}}_\rho(A;C |D)  \, ,
\end{equation}
with $f(\rho)$ a positive function depending only on the marginals of $\rho$ and explicitly defined in the proof. 
As a consequence, if $\rho$ is an approximate BS-QMC then it has a small $\Delta_\rho$. 
\end{proposition}

\begin{proof}

We use \cite[Proposition 17]{lucia2025spectral}, which shows
\begin{equation}
\Delta_\rho (A:C | D) \leq \frac{1}{2} \left( \norm{I- \rho_{AD} \rho_D^{-1} \rho_{CD} \rho_{ACD}^{-1} }_\infty + \norm{I- \rho_{ACD}^{-1} \rho_{AD} \rho_D^{-1} \rho_{CD}  }_\infty \right) \, .
\end{equation}
Let us perform a similar calculation to that of the first inequality in \cite[Remark 3.7]{gondolf2024conditional}. By Eq. \eqref{eq:DPIBSentropy}, identifying $\rho\equiv \rho_{AD}\otimes \tau_C$, $\sigma \equiv \rho_{ACD}$, $\mathcal{T}(\rho) \equiv \rho_D \otimes \tau_C$ and $\mathcal{T}(\sigma) \equiv \rho_{CD}$, we have
\begin{equation}
\left( \frac{\pi}{8}\right)^4 \frac{1}{d_C^6}\norm{\rho_{AD}^{-1/2} \rho_{ACD} \rho_{AD}^{-1/2}}_\infty^{-4} \norm{\rho_{D}^{-1}}^{-2}_\infty \norm{\rho_{AD} \rho_{D}^{-1} \rho_{CD} - \rho_{ACD}}_2^4 \leq \widehat{I}^{\operatorname{rev}}_\rho (A;C|D) \, ,
\end{equation}
which we can further lower bound by 
\begin{equation}
\underbrace{\left( \frac{\pi}{8}\right)^4 \frac{1}{d_C^6 d_{ABC}^2}\norm{\rho_{AD}^{-1/2} \rho_{ACD} \rho_{AD}^{-1/2}}_\infty^{-4} \norm{\rho_{D}^{-1}}_\infty^{-2}\norm{\rho_{ACD}^{-1}}_\infty^{-1} }_{g_1(\rho)} \norm{ I - \rho_{AD} \rho_{D}^{-1} \rho_{CD} \rho_{ACD}^{-1}}_1^4 \leq \widehat{I}^{\operatorname{rev}}_\rho (A;C|D) \, . 
\end{equation}
Taking $g(\rho) = g_1(\rho)^{-1}$, and noticing that the second term can be dealt with in an analogous way, we conclude the proof.
\end{proof}

Let us consider now an invertible state $\rho_{ABCD} \in \cS(\cH_{ABCD})$, its marginal $\rho_{ACD}$ and its associated $\eta_{ACD}=\frac{1}{d_D}\rho_D^{-1/2}\rho_{ACD} \rho_D^{-1/2}$. By \Cref{prop:relation_Irev_I}, if $[\eta_{AD},\eta_{CD}]=0$, then 
\begin{equation}
\widehat{I}^{\operatorname{rev}}_\rho(A;C |D)  \leq \text{min} \left\lbrace  g(\rho,d_A,d_C,d_D) I_\eta(A:C|D)^{1/4} ,   h(\rho,d_A,d_C,d_D) I_\eta(A:C|D)^{1/2}\right\rbrace \, .
\end{equation}
Therefore, as a consequence of \Cref{prop:relation_Irev_I} and \Cref{prop:Delta_BSQMC}, if $\eta$ has a small CMI, then $\Delta_\rho$ is small. This in particular allows us to transfer results such as the exponential decay of the CMI of Gibbs states of local Hamiltonians at any positive temperature from \cite{Kuwahara.2024} to the exponential decay of $\Delta_\rho$ (with respect to the distance between $A$ and $C$). Since this decay, for any construction such as that of \Cref{fig:2}, is sufficient to imply a positive spectral gap of the Davies Lindbladian with unique fixed point $\rho_{ABCD}$, any condition implying the decay of $\Delta_{\rho}$ is of great relevance.

There are two caveats in this approach though. The first one is that the decay of the CMI from \cite{Kuwahara.2024} includes a prefactor scaling exponentially on the sizes of $A$ and $B$; thus, this result can only give a correct decay for $\Delta_\rho$ if $A$ and $C$ are much smaller than the distance between them. The second one has to do with the physical significance of $\rho_{ACD}$ and $\eta_{ACD}$. In \cite{lucia2025spectral}, all states $\rho_{ABCD}$ considered are Gibbs states of local, commuting Hamiltonians. However, here we deal with $\rho_{ACD}$, which is not a Gibbs state of a local Hamiltonian unless the Hamiltonian in $ABCD$ has no correlations between $ACD$ and $B$ (which is rarely going to be the case). Thus, the physical structure of $\eta_{ACD}$ is unclear and this complicated obtaining information about its CMI.

\section{Conclusion}
In this work, we have established a connection between BS-QMCs $\rho_{ABC}$ and QMCs $\eta_{ABC}=\frac{1}{d_B}\rho_B^{-1/2}\rho_{ABC}\rho_B^{-1/2}$. Subsequently, we have used this connection to provide a structure theorem for BS-QMCs akin to the one in \cite{HaydenJozsaPetzWinter-StrongSubadditivity-2004}. Moreover, we have put forward a new recovery map $\Phi_{B \to AB}$ in the spirit of the Petz recovery map which satisfies a recovery condition if, and only if, the state is a BS-QMC, thereby solving an open question in \cite{gondolf2024conditional}. We have subsequently extended all these findings to the more general context of Petz- and BS-triples, finding a correspondence between them and the structural decomposition of states saturating DPIs for either the Umegaki or the BS-entropy. 

Furthermore, we have studied approximate BS-QMCs. In particular, our bounds show that for every approximate BS-QMC $\rho_{ABC}$, the associated $\eta_{ABC}$ is an approximate QMC.  An interesting open question is whether the converse is also true, because the converse in Proposition \ref{prop:relation_Irev_I} only holds if the marginals of $\eta_{ABC}$ commute. Instead of having an inequality between $\widehat{I}^{\operatorname{rev}}_{\rho}(A;C|B)$ and $I_\eta(A:C|B)$, it would so be interesting to have an inequality that relates $\widehat{I}^{x}_{\rho}(A;C|B)$ and $I_\rho(A:C|B)$ for $x \in \{\mathrm{os}, \mathrm{ts}, \mathrm{rev}\}$, since this quantifies in a way the difference between BS-QMCs and QMCs. In a similar vein, another interesting question is related to the discussion in Section \ref{sec:BS-QMC-not-QMC}: There, we   constructed a large family of explicit examples of BS-QMCs which are not  QMCs, but it would be interesting to quantify how much larger the set of BS-QMCs is compared to the set of QMCs which it contains, possibly quantifying this with their dimensions.

 Additionally, the aforementioned family of BS-QMCs which are not QMCs provides an interesting feature that is fundamentally quantum; namely, we showed that there is a regime for which the marginals in $AC$ of the BS-QMCs are entangled, whereas the corresponding QMCs to which they are mapped by the $\eta$ correspondence have always separable marginals between $A$ and $C$. It would be thus desirable to better understand this entanglement-breaking map (with converse creating entanglement), and particularly to verify whether the extension to $\eta$ and $\omega(X,V)$ between BS-triples and Petz-triples presents a similar behaviour, i.e. whether they are, when restricted to certain tests, entanglement-breaking (resp. creating). 

   Regarding the correspondence between Petz- and BS-triples, we have found that there is a complete correspondence in both directions when the map considered in the DPI is a conditional expectation. However, when considering channels, we show the equivalence between the triples when the BS-DPI is taken with respect to the channel and the Petz one with respect to the partial trace defining it. It is therefore an open question whether we can obtain a full equivalence, namely an analogue of \Cref{theo:theBeStTheorem}, for general quantum channels. Additionally, all our results based on the correspondence by the $\eta$ present the constraint from Petz to BS that $\eta$ on the second state has to be maximally mixed. It would be desirable to waive this restriction.  

Finally, we have applied our results to quantum spin chains and shown that the $\eta_{ABC}$ associated to Gibbs states $\rho_{ABC}$ of local, finite-range, translation-invariant Hamiltonians at any positive temperature exhibit superexponential decay. We have also given a bound for a quantity that, when exponentially-decaying with the distance between $A$ and $C$ for a state $\rho_{ABCD}$, gives a positive spectral gap for the Davies Lindbladian associated with unique fixed point $\rho_{ABCD}$. This leaves open the challenge to find more applications of BS-QMCs in quantum information theory and to endow them with a physical interpretation, in the way that QMCs correspond to Gibbs states of commuting, local Hamiltonians \cite{brown-2012}.

\vspace{1cm}

\textit{Acknowledgments:} The authors would like to thank Paul Gondolf, Milan Mosonyi, Samuel Scalet, and Michael M. Wolf for interesting discussions. A
part of this work was carried out during the BIRS-IMAG
workshop “Towards Infinite Dimension and Beyond
in Quantum Information” held at the Institute of
Mathematics of the University of Granada (IMAG) in
Spain. A.B.\ was supported by the French National Research Agency in the framework of the ``France 2030” program (ANR-11-LABX-0025-01) for the LabEx PERSYVAL. A.C.  acknowledge the support of the Deutsche Forschungsgemeinschaft (DFG, German Research Foundation) - Project-ID 470903074 - TRR 352, and also acknowledges funding by the Federal Ministry of Education and Research (BMBF) and the Baden-Württemberg Ministry of Science as part of the Excellence Strategy of the German Federal and State Governments. P.C.R. wants to thank the funding
by the Deutsche Forschungsgemeinschaft (DFG, German Research Foundation) under Germany’s
Excellence Strategy-EXC2111-390814868. This project was funded within the QuantERA II Programme which has received funding from the EU’s H2020 research and innovation programme under the GA No 101017733. A.J. was supported by the grant VEGA 2/0128/24 and by the Science and Technology Assistance Agency under the contract No. APVV-20-0069. 

\vspace{5pt}

\textbf{Conflict of interest:} The authors have no conflict of interest related to this publication.

\bibliographystyle{abbrv}
\bibliography{lit}

\end{document}